\documentclass[11pt]{article}
\usepackage[margin=1in]{geometry}

\usepackage{amsmath,amsthm,amssymb}
\usepackage{comment}
\usepackage{dsfont}
\usepackage{enumerate}
\usepackage{tikz}
\usepackage{hyperref}
\usepackage[capitalize]{cleveref}
\usepackage{setspace}
\usepackage{algpseudocode,algorithm,algorithmicx}
\usepackage{authblk}
\usepackage{pgfplots}

\newcommand{\commentout}[1]{}

\newcommand{\ba}{\boldsymbol{a}}
\newcommand{\bA}{\boldsymbol{A}}

\newcommand{\bB}{\boldsymbol{B}}
\newcommand{\bC}{\boldsymbol{C}}
\newcommand{\be}{\boldsymbol{e}}

\newcommand{\bg}{\boldsymbol{g}}

\newcommand{\bM}{\boldsymbol{M}}
\newcommand{\bq}{\boldsymbol{q}}
\newcommand{\bQ}{\boldsymbol{Q}}

\newcommand{\bR}{\boldsymbol{R}}
\newcommand{\bS}{\boldsymbol{S}}

\newcommand{\bu}{\boldsymbol{u}}
\newcommand{\bU}{\boldsymbol{U}}
\newcommand{\bv}{\boldsymbol{v}}
\newcommand{\bV}{\boldsymbol{V}}
\newcommand{\bw}{\boldsymbol{w}}
\newcommand{\bW}{\boldsymbol{W}}
\newcommand{\bx}{\boldsymbol{x}}
\newcommand{\bX}{\boldsymbol{X}}
\newcommand{\by}{\boldsymbol{y}}
\newcommand{\bY}{\boldsymbol{Y}}
\newcommand{\bz}{\boldsymbol{z}}
\newcommand{\bZ}{\boldsymbol{Z}}
\newcommand{\bgamma}{\boldsymbol{\gamma}}
\newcommand{\btheta}{\boldsymbol{\theta}}

\newcommand{\bxi}{\boldsymbol{\xi}}
\newcommand{\bPhi}{\boldsymbol{\Phi}}
\newcommand{\bPsi}{\boldsymbol{\Psi}}
\newcommand{\bXi}{\boldsymbol{\Xi}}
\newcommand{\bSigma}{\boldsymbol{\Sigma}}
\newcommand{\Cop}{C_{\rm op}}
\def\tauv{\boldsymbol{\tau}}

\newcommand{\SigmaX}{\boldsymbol{\Sigma}_{\boldsymbol{x}}}
\newcommand{\SigmaE}{\boldsymbol{\Sigma}_{\boldsymbol{e}}}
\newcommand{\SigmaY}{\boldsymbol{\Sigma}_{\boldsymbol{y}}}
\newcommand{\SigmaA}{\boldsymbol{\Sigma}_{\boldsymbol{a}}}

\newcommand{\0}{\boldsymbol{0}}
\newcommand{\id}{\boldsymbol{I}}
\renewcommand{\i}{\mathsf{i}}

\newcommand{\calC}{\mathcal{C}}

\newcommand{\calE}{\mathcal{E}}

\newcommand{\md}{{\rm md}}

\renewcommand{\subset}{\subseteq}
\renewcommand{\hat}{\widehat}
\renewcommand{\tilde}{\widetilde}
\renewcommand{\epsilon}{\varepsilon}
\renewcommand{\Re}{{\rm Re}}
\renewcommand{\Im}{{\rm Im}}
\def\sign{\mathrm{sign}}
\def\csign{\mathrm{sign}_\C}

\def\<{\big\langle}
\def\>{\big\rangle}
\def\({\Big(}
\def\){\Big)}
\def\calA{\mathcal{A}}

\def\C{\mathbb{C}}

\def\E{\mathbb{E}}

\def\F{\mathbb{F}}

\def\N{\mathbb{N}}
\def\calN{\mathcal{N}} 
\def\O{\mathbb{O}}
\def\P{\mathbb{P}}

\def\R{\mathbb{R}}

\def\calU{\mathcal{U}}

\def\T{\mathbb{T}}

\def\Z{\mathbb{Z}}

\newcommand{\tr}{\mathrm{tr}}

\newcommand{\SIGMA}{\boldsymbol{\Sigma}}

\newcommand{\SIGMADITHuniform}{\hat{\bSigma}^\square}
\newcommand{\SIGMADITHtriangular}{\hat{\bSigma}^\triangle}

\newcommand{\THETA}{\boldsymbol{\Theta}}
\newcommand{\andspace}{\quad{\rm and}\quad}

\newcommand{\pnorm}[2]{\left\| #1 \right\|_{#2}}

\newcommand{\curly}[1]{\left\{ #1 \right\} }

\DeclareMathOperator{\dist}{dist}
\DeclareMathOperator{\diag}{diag}
\DeclareMathOperator{\rank}{rank}

\newtheorem{theorem}{Theorem}[section]

\newtheorem{lemma}[theorem]{Lemma}
\newtheorem{corollary}[theorem]{Corollary}
\theoremstyle{definition}

\newtheorem{remark}[theorem]{Remark}
\newtheorem{assumption}[theorem]{Assumption}

\numberwithin{equation}{section}

\begin{document}

\title{Subspace and DOA estimation under coarse quantization}

\author[$\star$]{Sjoerd Dirksen}
\author[$\circ$]{Weilin Li}
\author[$\dagger$,$\ddagger$]{Johannes Maly}
\affil[$\star$]{\normalsize Mathematical Institute, Utrecht University, The Netherlands}
\affil[$\circ$]{Department of Mathematics, City College of New York, USA}
\affil[$\dagger$]{Department of Mathematics, LMU Munich, Germany}
\affil[$\ddagger$]{Munich Center for Machine Learning (MCML)}

\date{}
\maketitle

\begin{abstract}
\noindent We study direction-of-arrival (DOA) estimation from coarsely quantized data. {We focus on a two-step approach which first estimates the signal subspace via covariance estimation and then extracts DOA angles by the ESPRIT algorithm. In particular,} we analyze two stochastic quantization schemes which use dithering: a one-bit quantizer combined with rectangular dither and a multi-bit quantizer with triangular dither. For each quantizer, we derive rigorous high probability bounds for the distances between the true and estimated signal subspaces and DOA angles. Using our analysis, we identify scenarios in which subspace and DOA estimation via {triangular dithering} qualitatively outperforms rectangular dithering. We verify in numerical simulations that our estimates are optimal in their dependence on the smallest non-zero eigenvalue of the target matrix. {The resulting subspace estimation guarantees are equally applicable} in the analysis of other spectral estimation algorithms and related problems.
\end{abstract}

\section{Introduction}
\label{sec:Introduction}

The question of how to estimate the direction-of-arrivals (DOAs) {of incoming signals} is decades old \cite{paulraj1993subspace}, but still of great interest to modern applications, e.g., wireless telecommunication \cite{ruan2022doa}. Roughly speaking, it is a nonlinear inverse problem of detecting the angles $\btheta=(\theta_1,\dots,\theta_s)$ between a receiver and a collection of $s$ objects, given access to $p$ many Fourier samples collected per time, over a period of $n$ times. MUSIC \cite{schmidt1986multiple} and ESPRIT \cite{kailath1989esprit} are two foundational algorithms that are widely used for DOA (and other spectral estimation problems), and are still considered state of the art \cite{li2022stability}.

The first step in these algorithms is to approximate the underlying signal subspace: Given unknown signals $\bx_1,\dots,\bx_n \in \C^p$ that are concentrated around an $s$-dimensional subspace $\bU \subset \C^p$, e.g., complex $p$-dimensional Gaussians with covariance matrix of rank $s$ and leading eigenspace $\bU$, 
estimate $\bU$ from noisy observations
\begin{align}
\label{eq:NoisySamples}
    \by_k = \bx_k + \be_k, \qquad k \in [n],
\end{align}
where the vectors $\be_k \in \C^p$ model random noise. In this paper, we concentrate on the setting where the entries of $\be_k$ are i.i.d. samples from a fixed subgaussian distribution with mean zero and variance $\nu^2$, e.g., a complex Gaussian distribution.  

Defining the observation matrix $\bY_n = [\by_1 \cdots \by_n] \in \C^{p\times n}$, a natural estimator $\widehat\bU$ of $\bU$ is obtained by computing the left singular vectors of a truncated SVD of $\bY_n$, or equivalently, the leading eigenvectors of $\frac 1n\bY_n\bY_n^*$\footnote{Note that $\frac{1}{n} \bY_n\bY_n^*$ is the sample covariance matrix of the samples $\by_1,\dots,\by_n$.}. The estimator $\hat\bU$ is consistent because the calculation
$$
\E (\by_k\by_k^*) = \E (\bx_k\bx_k^*) + \nu^2 \id_p
$$
shows that in expectation, the noise shifts the eigenvalues of $\frac 1n\bY_n\bY_n^*$, but not {their eigenspaces}.\footnote{For heterogeneous noise, e.g., if the entries of $\be_k$ are drawn from different distributions, $\hat\bU$ would no longer be a consistent estimator for $\bU$ and a different estimator is required. We stick with homogeneous noise since it is a natural starting point for quantized DOA estimation. \label{footnote:HeteroNoise}} In the finite sample setting, the performance of $\widehat\bU$ can be measured in terms of a suitable subspace angle distance $\dist(\widehat\bU,\bU)$, see \eqref{eqn:distSubONB}. Combining well-known results on covariance estimation \cite{koltchinskii2017concentration} with the Davis-Kahan theorem \cite{davis1970rotation} then leads to non-asymptotic error bounds for $\dist(\widehat\bU,\bU)$ in terms of the oversampling ratio $n/s$. 
Interestingly, if the noise $\be_k$ is drawn from certain distributions, a superior error bound for $\dist(\widehat\bU,\bU)$ was derived in \cite{cai2018rate}.\footnote{One intuitive and non-rigorous explanation for this phenomenon is that it is prohibitively unlikely for {generic} random noise to adversarially perturb a low dimensional subspace embedded in a high dimensional ambient space.} This improved rate was exploited in \cite{li2022stability}, which also showed how to relate $\dist(\widehat\bU,\bU)$ to the error between the true arrival angles $\btheta$ versus the estimated ones $\hat\btheta$ computed from ESPRIT or MUSIC.

Since applications of DOA estimation normally lie at the interface of analog and digital domain, an additional challenge is to take into account the impact of analog-to-digital conversion, i.e., quantization \cite{gray1998quantization} of observations. In this case one has only access to quantized observations $Q(\by_k)$ for $k\in [n]$, instead of the analog $\by_k$. {Note however that one is in general allowed} to design the {\it quantizer} $Q$, {taking into account potential hardware limitations}. We restrict ourselves here to \emph{memoryless scalar quantization}, i.e., $Q \colon \C \to \calA$ is a univariate function from $\C$ to a finite alphabet $\calA \subset \C$ that is applied entry-wise to vectors in $\C^p$. Especially in view of modern large-scale applications, there has been growing interest in \emph{coarse quantization schemes}, where the number of bits per scalar, $\log_2(|\calA|)$, is small.

A notable line of works (\cite{BaW02,yu2016doa,HuL19,LSWW19,wei2020gridless,zhou2020direction,sedighi2021performance,li2024rocs}) examines DOA estimation from samples that are quantized by the simplest coarse quantizer $Q(\cdot) = \sign(\cdot)$ or its complex-valued two-bit counterpart, often relying on the so-called arcsine law \cite{van1966spectrum}. {On the one hand, the $\mathrm{sign}$-quantizer is cheap to implement. On the other hand, it loses any scaling information and thus allows reliable recovery only under restrictive assumptions on the signals and the noise.} A second, more recent line of works \cite{PZXD24,YEES24} mitigates the shortcomings of this simple quantizing scheme by considering dithered sign-quantizers of the shape $Q(\cdot) = \sign(\cdot + \tau)$ where the dither $\tau$ is intentionally added, well-designed noise. Whereas \cite{PZXD24} examines Gaussian noise as dither, the authors of \cite{YEES24} follow the ideas in \cite{dirksen2022covariance} and use uniform dithering noise together with two-bit observations per sample. {One aspect of the present work is that quantizers with triangular dither seem to be natural for DOA estimation because they are} able to create favorable statistical properties for the {\it quantization noise} $\bx_k-Q(\by_k)$ \cite{gray1993dithered} that are necessary to access improved bounds for subspace estimation \cite{cai2018rate} and MUSIC/ESPRIT \cite{li2022stability}.
A more detailed review of related work on quantized DOA estimation and dithered quantizers is provided in Section \ref{sec:RelatedWork}.

\subsection{Contribution and Outline}

In this work, we analyze subspace identification and the performance of ESPRIT for spectral estimation from samples that are collected via a dithered quantizer. In particular,
\begin{enumerate}[(i)]
    \item we provide high-probability, non-asymptotic bounds for $\dist(\widehat\bU,\bU)$ when $\widehat\bU$ is obtained from quantized observations following the quantization model with uniformly distributed dither as proposed in \cite{dirksen2022covariance,YEES24}, see Theorem \ref{thm:subspacerectangle}. {Due to the characteristic shape of the dither density, we refer to this setting as \emph{rectangular dithering}, cf.\ Figure \ref{fig:Dither}.}
    \item we show that, for low measurement noise and higher bit-rates, \emph{triangular dithering} \cite{gray1993dithered} yields superior error bounds on $\dist(\widehat\bU,\bU)$, see Theorem \ref{thm:subspacerectangletriangle}.
    Here, the dither $\tau$ is a sum of two independent copies of a uniform random variable, and the name triangular dithering again stems from the characteristic shape of the dither density, cf.\ Figure \ref{fig:Dither}.
    {\textit{The superior performance of triangular dithering in subspace estimation is particularly surprising since it stands in stark contrast to recently derived covariance estimation bounds.} Indeed, for covariance estimation the operator norm error bounds under rectangular dithering \cite{dirksen2022covariance} and triangular dithering \cite{chen2022quantizing} are comparable.} {As we demonstrate here, the latter method generates quantization noise that has favorable statistical properties for subspace estimation.} {For the herein considered subspace estimation this highlights the necessity of following two different strategies in the two dithering regimes in order to obtain optimal bounds.}
    \item we combine the obtained subspace identification bounds with the results in \cite{li2022stability} to derive novel performance bounds for ESPRIT, see \cref{thm:espritrectangular,cor:espritrectangular,thm:esprittriangular,cor:esprittriangular}.
    \item we provide extensive numerical simulations which verify the theoretical error bounds and demonstrate that they are sharp, see Section \ref{sec:Numerics}.
    \item in deriving the above results, we strengthen the stochastic subspace perturbation estimate from \cite[Theorem 3]{cai2018rate} {to a high-probability bound under relaxed assumptions, see \cref{thm:CZvarHP}.}
\end{enumerate}

To the best of our knowledge this work is the first to consider triangular dithering in the context of DOA estimation. It furthermore provides the first non-asymptotic error bounds on subspace/DOA estimation from dithered quantized samples that detail the influence of the conditioning number of the signal covariance matrix, cf.\ our discussion in Section \ref{sec:RelatedWork}.
Note that the results in (i), (ii), and (v) are of independent interest. Other applications that include subspace identification as a core component encompass latent factor analysis, matrix denoising, and spectral clustering, see \cite{chen2020spectral}. Furthermore, we emphasize that we restrict ourselves to ESPRIT for the sake of conciseness. Combining the results in (i) and (ii) with \cite{li2022stability} also provides performance guarantees for MUSIC when applied to spectral estimation from quantized samples. Table \ref{table:Results} provides an overview over the core results.

\begin{table}[]
\begin{center}
\caption{Overview over Results}
\label{table:Results}
\begin{tabular}{||c | c || c | c||} 
 \hline
 Result & Reference & Result & Reference \\ [0.5ex] 
 \hline\hline
 Subspace estimation (rectangular) & Thm.\ \ref{thm:subspacerectangle} & DOA estimation (rectangular) & Thm.\ \ref{thm:espritrectangular} \\ 
 \hline
 Subspace estimation (triangular) & Thm.\ \ref{thm:subspacerectangletriangle} & DOA estimation (triangular) & Thm.\ \ref{thm:esprittriangular} \\
 \hline
 \begin{tabular}{@{}c@{}} Statistical properties \\ of quantization noise \end{tabular} & Lem.\ \ref{lem:triquannoise} & \begin{tabular}{@{}c@{}}DOA estimation (rectangular) \\ well-separated case \end{tabular} & Cor.\ \ref{cor:espritrectangular} \\
 \hline
 \begin{tabular}{@{}c@{}} Strengthened version \\ of \cite[Theorem 3]{cai2018rate} \end{tabular} & Thm.\ \ref{thm:CZvarHP} & \begin{tabular}{@{}c@{}}DOA estimation (triangular) \\ well-separated case \end{tabular} & Cor.\ \ref{cor:esprittriangular} \\
 \hline
\end{tabular}
\end{center}
\end{table}

The paper is organized as follows. In Sections \ref{sec:RelatedWork} and \ref{sec:Notation} we discuss related work in more detail and introduce the notation that is used in the remaining paper.
In Section \ref{sec:MainResults} we then present our main results on subspace identification from quantized samples. The developed theory is applied to ESPRIT and DOA estimation in Section \ref{sec:ESPRITforDOA} and numerical experiments supporting our claims are provided in Section \ref{sec:Numerics}. {The proofs can be found in Section \ref{sec:Proofs}.}

\subsection{Related work}
\label{sec:RelatedWork}

While a thorough review of the by now rich literature on DOA estimation from quantized samples (\cite{BaW02,yu2016doa,HuL19,LSWW19,wei2020gridless,zhou2020direction,sedighi2021performance,li2024rocs}) is beyond the scope of this paper, let us describe some existing approaches. In \cite{BaW02}, DOA estimation from one-bit samples is performed by one-bit covariance estimation via the arcsine law. Due to near-linearity of the arcsine function close to the origin, in \cite{HuL19} it is argued that for high SNRs the sample covariance matrix of the one-bit samples can be used without further adaption. The authors of \cite{huang2019direction} propose to first recover the unquantized measurements and then apply MUSIC. In \cite{yoffe2019direction}, a maximum likelihood-based approach is examined.
On the one hand, the simplicity of the sign-quantizer is appealing from a hardware perspective. On the other hand, it comes with strong limitations like scaling invariance. A second, more recent line of works (\cite{PZXD24,YEES24}) mitigates these shortcomings by considering dithered sign-quantizers of the form \eqref{eq:QuantizedSnapShots} below. Finally, in \cite{sedighi2021doa} the authors of \cite{sedighi2021performance} extend their ideas from one-bit to multi-bit quantizers.

Dithering, i.e., adding well-designed noise to a signal before quantization, has a long history in signal processing \cite{Rob62}. In particular, the last decade showed substantial progress in deriving rigorous non-asymptotic performance guarantees for signal reconstruction from one-bit measurements, see e.g. \cite{BJK15,Dir19} and the references therein. Concrete examples encompass one-bit matrix completion (e.g., \cite{CaZ13,DPB14,chen2022high,EYS23}), reconstructing a signal in an unlimited sampling framework with one-bit quantization (e.g., \cite{EMY22,EMY23}), and one-bit quadratic sensing problems such as phase retrieval (e.g., \cite{EYN23,EYS22}). Most relevant to this work is that only recently the first non-asymptotic (and near-minimax optimal) guarantees for estimating covariance matrices from one-bit samples with uniform dither have been derived \cite{dirksen2022covariance}. In \cite{yang2023plug}, the results of \cite{dirksen2022covariance} have been generalized to the complex domain and applied to massive MIMO; in \cite{dirksen2024tuning} a data-adaptive variant of the estimator in \cite{dirksen2022covariance} has been developed. A recent work \cite{chen2022high} modified the strategy of \cite{dirksen2022covariance} to cover heavy-tailed distributions by using truncation before quantizing. The work \cite{chen2023parameter} introduced the idea of using triangular rather than uniform dithering in the context of covariance estimation. 

Almost all previously mentioned studies on DOA estimation from coarsely quantized samples are of empirical nature, proposing algorithmic approaches to the problem and evaluating their performance in simulations. The only exceptions are \cite{sedighi2021performance}, \cite{YEES24}, and \cite{gunturk2022quantization}. The authors of \cite{sedighi2021performance} focus on one-bit DOA estimation via Sparse Linear Arrays. They provide conditions under which the identifiability of source DOAs from unquantized data is equivalent to the one of one-bit data. Furthermore, they provide a Cramér-Rao bound analysis and a MUSIC-based reconstruction approach with asymptotic error guarantees. The authors of \cite{YEES24} build upon the idea of \cite{xiao2020deepfpc} in which DOA estimation from undithered one-bit samples is performed by Learned Iterative Soft-Thresholding (LISTA), i.e., unfolding and training ISTA as a network. By adding a uniform dither to the quantization model and relying on the theoretical analysis of the corresponding quantized covariance estimator in \cite{yang2023plug}, they can derive performance guarantees for the one-bit LISTA approach from the results in \cite{chen2018theoretical}. The authors of \cite{gunturk2022quantization} used ESPRIT for the quantized single-snapshot DOA problem, where information is only collected at one time instance. This setting cannot be treated using statistical methods. Instead, they exploited analytic properties of the Fourier transform by using a two-bit beta-quantization method, which is very different from dithered quantization.

\subsection{Notation}
\label{sec:Notation}

Before presenting our main results, it will be convenient to discuss the notation that is used throughout this work. We write $[n] = \{ 1,...,n \}$ for $n\in \mathbb{N}$ and let $\F \in \{ \R, \C\}$ be a placeholder for any of the two fields. We use the notation $a \lesssim_{\alpha} b$ (resp.\ $\gtrsim_{\alpha}$) to abbreviate $a \le C_{\alpha}b$ (resp.\ $\ge$), for a constant $C_{\alpha} > 0$ depending only on $\alpha$. Similarly, we write $a \lesssim b$ if $a \le Cb$ for an absolute constant $C>0$. We write $a\simeq b$ if both $a\lesssim b$ and $b\lesssim a$ hold (with possibly different implicit constants). Whenever we use absolute constants $c,C > 0$, their values may vary from line to line. 

\textbf{Scalar valued functions.} We let scalar-valued functions act entry-wise on vectors and matrices. In particular, the real-valued sign function is given by
\begin{align*}{}
[\sign_\R(\bx)]_i = \begin{cases}
1 & \text{if } x_i\geq 0 \\
-1& \text{if } x_i<0,
\end{cases}
\end{align*}
for all $\bx\in \R^p$ and $i\in [p]$. We furthermore define the complex-valued sign-function by 
\begin{align*}
    \csign(\bz) = \sign_\R(\Re(\bz)) + \i\, \sign_\R(\Im(\bz)) \in \{ \pm 1 \pm \i \}^p, 
\end{align*}
for any $\bz \in \C^p$. Note that $\sign(\bx)_\R \ne \sign_\C(\bx)$, for $\bx \in \R^p$ (since for instance $\sign_\R(0)\not=\sign_\C(0)$). Whenever we make a statement regarding the space $\F^p$ where $\F \in \{\R,\C\}$, we use $\sign_\F$ to refer to the respective sign-function. We abbreviate the squared modulus of the sign function by $c_\F := |\sign_\F(\cdot)|^2$, i.e., $c_\R = 1$ and $c_\C = 2$.

\textbf{Matrix calculus.} We denote the identity in $\R^{p\times p}$ by $\id_p$, the all ones-matrix by $\boldsymbol{1}_p$, and the zero matrix by $\0_p$. If the dimension is clear from the context, we omit the subscript $p$. For $\bZ \in \F^{p\times p}$, we denote the operator norm by $\| \bZ \| = \sup_{\bu \in \mathbb{S}^{p-1}} \| \bZ\bu \|_2$ (i.e., it is the maximum singular value, where $\mathbb{S}^{p-1}$ is the set of vectors in $\F^p$ of unit norm), the entry-wise max-norm by $\| \bZ \|_{\infty} = \max_{i,j} |Z_{i,j}|$, and the $\ell^1$ entry-wise norm by $\|\bZ\|_1=\sum_{i,j} |Z_{i,j}|$.
The $\diag$-operator, when applied to a matrix, extracts the diagonal as a vector; when applied to a vector, it outputs the corresponding diagonal matrix. 
For $m \ge n$, we let $\O^{m\times n}$ be the set of $m\times n$ matrices over $\F$ whose columns are orthonormal.  
For a Hermitian matrix $\bA\in \C^{m\times m}$, we let $\lambda_1(\bA)\geq \lambda_2(\bA)\geq \cdots \geq \lambda_m(\bA)$ denote its eigenvalues in non-increasing order. We say $\bU\in \O^{m\times r}$ denotes an orthonormal basis for the leading eigenspace of $\bA$ if the $r$ columns of $\bU$ are eigenvectors of $\bA$ that correspond to the largest $r$ eigenvalues of $\bA$. If $\lambda_r(\bA)=\lambda_{r+1}(\bA)$, there is additional ambiguity in the choice $\bU$, and in this case, $\bU$ may correspond to any valid leading eigenspace of $\bA$. For a general matrix $\bB\in \C^{m\times n}$, we let $\sigma_k(\bB):=\sqrt{\lambda_k(\bB^*\bB)}$ be its $k$-th largest singular value. Its {$r$-th condition number is denoted by $\kappa_r(\bB):=\sigma_1(\bB)/\sigma_r(\bB) \in [1,\infty]$.} 
For Hermitian $\bW,\bZ\in \C^{p\times p}$ we write $\bW \preceq \bZ$ if $\bZ-\bW$ is positive semidefinite.

\textbf{Probability.} We denote a real Gaussian random vector with mean $\boldsymbol{\mu} \in \R^p$ and covariance matrix $\bSigma \in \R^{p\times p}$ by $\bx \sim \calN(\boldsymbol{\mu},\bSigma)$ and a circularly symmetric complex Gaussian random vector with covariance matrix $\bSigma \in \C^{p\times p}$ by $\bz \sim \calC\calN(\boldsymbol{0},\bSigma)$. If $\boldsymbol{\mu} = \boldsymbol{0}$ and $\bSigma = \id_p$, we call $\bx$ resp.\ $\bz$ (complex) standard normal. The subexponential ($\psi_1$-) and subgaussian ($\psi_2$-)norms of a random variable $X$ are defined by 
\begin{align*}
\pnorm{X}{\psi_1} = \inf \curly{ t>0 \colon \E \, e^{|X|/t } \le 2 }.
\end{align*}
and
\begin{align*}
\pnorm{X}{\psi_2} = \inf \curly{ t>0 \colon \E \, e^{|X|^2/t^2 } \le 2 }.
\end{align*}
A mean-zero random vector $\by$ in $\R^p$ is called $K$-subgaussian if 
$$\|\langle \by,\bx \rangle\|_{\psi_2} \leq K \| \langle \by,\bx\rangle \|_{L_2} = K (\E \langle \by,\bx\rangle^2)^\frac{1}{2}  \quad \mbox{ for all }  \bx \in \R^p.$$
If $\by$ is isotropic, i.e., $\E(\by\by^*)=\id_p$, then the latter is equivalent to  
$$\|\langle \by,\bx \rangle\|_{\psi_2} \leq K\|\bx\|_2 \quad \mbox{ for all }  \bx \in \R^p.$$
We call a random vector $\by \in \C^p$ $K$-subgaussian if both $\Re(\by)$ and $\Im(\by)$ are subgaussian. 
For $p \ge 1$, we denote by $\mathcal{U}_\R[a,b]^p$ the uniform distribution on the rectangle $[a,b]^p$ and by $\mathcal{U}_\C[a,b]^p$ the uniform distribution on the set $[a,b]^p + \i [a,b]^p$, i.e., $X \sim \mathcal{U}_\C[a,b]^p$ if and only if $\Re(X) \sim \mathcal{U}_\R[a,b]^p$ and $\Im(X) \sim \mathcal{U}_\R[a,b]^p$ are independent.
Finally, we will use that
\begin{align} \label{eq:Kadison}
    \mathbb{E}\bZ^*\mathbb{E}\bZ \preceq \mathbb{E}(\bZ^*\bZ) \quad \text{for any } \bZ\in \C^{p_1\times p_2},
\end{align}
which is immediate from 
$$\0 \preceq \mathbb{E}[(\bZ-\mathbb{E}\bZ)^*(\bZ-\mathbb{E}\bZ)].$$

\section{Subspace identification from quantized observations}
\label{sec:MainResults}

The first step of ESPRIT consists in identifying the dominant subspace of the sample distribution. We will thus initially focus on the subproblem of subspace identification from quantized samples.
Given the observation model in \eqref{eq:NoisySamples}, we will work with the following set of assumptions which is used in various applications.

\begin{assumption} \label{assump:main}
    Fix $n,p\in\N$.
    \begin{enumerate}
        \item[(i)] Let $\bx_1,\dots,\bx_n\in \F^p$ be deterministic (d) or stochastic (s) vectors. We then either 
        \begin{enumerate}
            \item[(d)] set $\SigmaX:=\frac 1 n \sum_{k=1}^n \bx_k \bx_k^*$ or
            \item[(s)] assume $\bx_1,\dots,\bx_n\overset{\text{i.i.d.}}{\sim} \bx$, where $\bx\in\F^p$ is $K$-subgaussian, and set $\SigmaX:=\E (\bx \bx^*)$.
        \end{enumerate}
        \item[(ii)] Let $\be_1,\dots,\be_n\overset{\text{i.i.d.}}{\sim} \be$, where $\be\in\F^p$ is $K$-subgaussian with {uncorrelated} entries that are mean-zero and have variance $\nu^2$. If $\F=\C$, we assume that the real and imaginary parts of the entries of $\be$ are independent.\footnote{Our assumption implies that $\SigmaE=\nu^2 \id$. This is an important property as it implies that the eigenspaces of $\SigmaX$ and $\SigmaY=\SigmaX+\SigmaE$ are identical, see also the remark in Footnote \ref{footnote:HeteroNoise}.}
        \item[(iii)] Let $\bx_1,\dots,\bx_n,\be_1,\dots,\be_n$ be independent. 
    \end{enumerate} 
\end{assumption}

We define $\SigmaY := \SigmaX + \SigmaE$ with $\SigmaE:= \E (\be\be^*)$.
Here, each $\bx_k$ represents the true empirical observation, $\be_k$ plays the role of (unknown) stochastic noise, and $\by_k$ represents the observed data before quantization. Assumption \ref{assump:main}(ii) implies that $\SigmaE=\nu^2 \id$. This is an important property as it implies that the eigenspaces of $\SigmaX$ and $\SigmaY=\SigmaX+\SigmaE$ are identical, see also the remark in Footnote \ref{footnote:HeteroNoise}.

To measure the distance between subspaces, we will rely on the sine-theta distance \cite{davis1970rotation} which is defined as follows. For two subspaces of the same dimension embedded in $\F^p$, let $\bU$ and $\bV$ be matrices whose columns form orthonormal bases for these spaces. We define $\sin(\bU,\bV)$ to be the diagonal matrix consisting of sine of the principle angles between $\bU$ and $\bV$. The sine-theta distance $\dist(\bU,\bV)$ is then defined as the sine of the largest principle angle between these spaces. It can also be expressed as  
\begin{equation}
\label{eqn:distSubONB}
\dist(\bU,\bV)= \|\sin(\bU,\bV)\|= \|\bU \bU^*-\bV\bV^*\|.
\end{equation}

\subsection{Rectangular dithering}
\label{sec:rectangular}

Let us begin our study with the quantization model considered in \cite{dirksen2022covariance,YEES24}, i.e., a sign-quantizer with uniform dithers. Given observations $\by_k\in \F^p$ as in \eqref{eq:NoisySamples}, we thus collect quantized samples $\bq_k^\square$ and $\dot \bq_k^\square$ where  
\begin{align} 
\label{eq:QuantizedSnapShots}
    &\left\{ \bq_k^\square, \dot \bq_k^\square \right\} 
    := \left\{ \sign_\F \big( \by_k + \tauv_k^\square \big), \, 
    \sign_\F \big(\by_k + \dot \tauv_k^\square \big) \right\}
\end{align}
and the {\it dithering vectors} $\tauv_k^\square$ and $\dot \tauv_k^\square$ are independently drawn from $\mathcal{U}_\F[-\lambda,\lambda]^p$, for $\lambda > 0$ to be determined later. In the case of real measurements, each entry of $\tauv_k^\square, \dot \tauv_k^\square$ is uniformly distributed in $[-\lambda,\lambda]$ and $\bq_k^\square, \dot \bq_k^\square \in \{\pm 1\}^p$, i.e., $\calA_\R=\{\pm 1\}$ is a one bit alphabet; in the complex case, the real and complex part of each entry of $\tauv_k^\square, \dot \tauv_k^\square$ are independently drawn from $\mathcal{U}[-\lambda,\lambda]$ and $\bq_k^\square, \dot \bq_k^\square \in \{\pm 1\pm \i \}^p$, i.e., $\calA_\C=\{\pm 1\pm \i\}$ is a two bit alphabet. Due to the shape of the probability density of the dither, we refer to this model as \emph{rectangular dithering} {and highlight related quantities with a $\square$-symbol in the superscript}.

\begin{remark}
\label{rem:ADC}
    {Whereas the quantization model in \eqref{eq:QuantizedSnapShots} requires four bits of information per complex-valued observation, we highlight that it can be implemented in hardware with two one-bit ADC units, see \cite[Figure 2]{yang2023plug}.}
\end{remark}

As mentioned in Section \ref{sec:Introduction}, a generic strategy for subspace identification is to approximate the covariance matrix of the samples $\SigmaY$ by the sample covariance matrix $\frac{1}{n} \bY_n\bY_n^*$ and then to use the following estimate, which is a straightforward consequence of the Davis-Kahan {sin-theta} theorem \cite{davis1970rotation}. 

\begin{lemma}
\label{lem:PSestviaCovEst}
    Let $\hat{\SIGMA}\in \mathbb{C}^{p\times p}$ be Hermitian. Let $\bU,\hat{\bU} \in \O^{p\times s}$ denote orthonormal bases of the leading eigenspaces of $\SigmaX$ and $\hat{\SIGMA}$, respectively. If $\lambda_s(\SigmaX)>\lambda_{s+1}(\SigmaX)$,\footnote{Assuming an eigenvalue gap $\lambda_s(\SigmaX)>\lambda_{s+1}(\SigmaX)$ is necessary to discuss approximation of the leading $s$-dimensional eigenspace, which would otherwise be ill defined.} then
    \begin{equation}
        \label{eq:wedin2}
        \dist(\hat{\bU},\bU)
        \leq \min\left\{ \, 1, \, \frac{(1+\sqrt 2) \, \|\hat{\SIGMA}-\SIGMA_{\by}\|}{\lambda_s(\SigmaX)-\lambda_{s+1}(\SigmaX)} \, \right\}. 
    \end{equation}
\end{lemma}

The proof of Lemma \ref{lem:PSestviaCovEst} is provided in Section \ref{sec:PSestviaCovEst}.
Recalling the ideas in \cite{dirksen2022covariance,yang2023plug}, a natural estimator for $\SIGMA_{\by}$ from quantized samples as in \eqref{eq:QuantizedSnapShots} is given by $\SIGMADITHuniform_n$ where
\begin{align} \label{eq:TwoBitEstimator}
    \SIGMADITHuniform_n = \frac{1}{2}\hat{\SIGMA}'_n + \frac{1}{2}(\hat{\SIGMA}'_n)^*
\end{align}
and
\begin{align} \label{eq:AsymmetricEstimator}
    \hat{\SIGMA}'_n = \frac{\lambda^2}{n} \sum_{k=1}^n \bq_k^\square (\dot \bq_k^\square)^*.
\end{align}
By slightly modifying the bounds on $\| \SIGMADITHuniform_n - \SigmaY \|$ derived in \cite{dirksen2022covariance,yang2023plug} and combining the resulting estimates with Lemma \ref{lem:PSestviaCovEst}, we obtain the following result.
Before providing the statement, we define the constants
\begin{equation}
	\label{eq:Cy_defs}
    \begin{split}
    C_{\infty} 
    &= \max_{k \in [n]} \| \E(\by_k\by_k^*) \|_\infty \\
    C_{\text{op}} 
    &= \max_{k \in [n]} \| \E(\by_k\by_k^*) \|.    
    \end{split}
\end{equation} 
 
\begin{theorem}
	\label{thm:subspacerectangle}
	Suppose \cref{assump:main} holds and we have an eigengap $\lambda_s(\SigmaX)>\lambda_{s+1}(\SigmaX)$. Then, there exist constants $B_K,C_K>0$ that depend only on $K$ such that the following hold. For any $\lambda^2 \geq B_K \log(n) C_{\infty}$, let $\bU$, $\bU_n^\square \in \O^{p\times s}$ denote orthonormal bases of the leading eigenspaces of $\SigmaX$ and $\SIGMADITHuniform_n$. We have with probability at least $1-e^{-t}$ that 
	\begin{equation}
    \label{eq:subspace}
    \begin{split}
    &\dist(\bU^\square_n,\bU) \\
	  &\leq \min \left\{1, \, \frac{C_K \, (\Cop^{1/2} + \lambda)}{\lambda_s(\SigmaX)-\lambda_{s+1}(\SigmaX)} \sqrt{\frac{p (\log(p)+t)}{n}}
	    + \frac{C_K \lambda^2}{\lambda_s(\SigmaX)-\lambda_{s+1}(\SigmaX)}  \frac{p (\log(p) + t)}{n} \right\}. 
    \end{split}
	\end{equation}
    {Moreover, if} $\bx$ and $\be$ are bounded and $\lambda^2 \ge C_{\infty}$ a.s., then \eqref{eq:subspace} holds with probability at least $1-e^{-t}$. 
\end{theorem}

Theorem \ref{thm:subspacerectangle} follows from combining Lemma~\ref{lem:PSestviaCovEst} with Theorem~\ref{thm:CE_ErrorBound} in Section \ref{sec:subspace}.

\begin{remark}
\label{rem:Cop}
	Note that {in Theorem \ref{thm:subspacerectangle}} the dependence on the noise level is hidden in the parameters $C_{\infty}$ and $C_{\text{op}}$. Indeed, given Assumption \ref{assump:main} we have in the case of random samples, 
	\begin{align*}
	C_{\infty} &= \| \SigmaY \|_\infty = \|\SigmaX\|_{\infty}+\nu^2, \\
	C_{\text{op}} &= \| \SigmaY \| = \|\SigmaX\|+\nu^2,
	\end{align*}
	whereas for deterministic samples,
	\begin{align*}
	C_{\infty} &= \max_{k\in [n]} \|\bx_k\|_{\infty}^2 + \nu^2, \\
	C_{\text{op}}  &= \max_{k\in [n]} \|\bx_k\|_2^2 + \nu^2. 
	\end{align*}
\end{remark}

On the one hand, the generic approach of Theorem \ref{thm:subspacerectangle} is appealing due to its straightforward derivation from existing insights on the minimax-optimal estimator in \eqref{eq:TwoBitEstimator}. On the other hand, the associated sample complexity estimate produced by Lemma~\ref{lem:PSestviaCovEst} {may be suboptimal}. Let us make this more precise. While $\SIGMADITHuniform_n$ will achieve the minimax-optimal spectral norm estimation error of order $1/\sqrt{n}$, the resulting subspace identification bound in Theorem \ref{thm:subspacerectangle} requires $n\sim \lambda_s^{-2}(\SigmaX)$ samples to yield a non-trivial eigenspace estimate. In the following, we will see that a better error dependency on the condition number of $\SigmaX$ can be achieved by combining a uniform scalar (multi-bit) quantizer with triangular dithering. As we verify in our numerical experiments in \cref{sec:Numerics} below, {both Theorem \ref{thm:subspacerectangle} and the improved scaling for triangular dithering are sharp in their dependence on $\lambda_s(\SigmaX)$}.

\subsection{Triangular dithering}

\label{sec:triangular}

In recent works on covariance estimation from quantized samples (\cite{chen2022quantizing,chen2023parameter}) it is suggested to combine a uniform infinite-range quantizer 
\begin{equation}
    \begin{split}
	Q_\mu &\colon \R \to \calA_\mu, \\
    \calA_\mu&:= 2\mu \mathbb Z + \mu = \{ \dots, -3\mu, -\mu, \mu, 3\mu, \dots \} \\
	Q_\mu(x) &:= 2\mu \left( \left\lfloor \frac{x}{2\mu} \right\rfloor + \frac{1}{2} \right)
    \end{split} \label{eq:Qmudef}
\end{equation}
of resolution $\mu > 0$ with triangular dithering. To be more precise, given samples $\by_1,\dots,\by_n \in \R^p$ they apply $Q_\mu$ entry-wise and collect quantized samples
\begin{align*}
	\bq_k^\triangle = Q_\mu(\by_k + \tauv_k^\triangle) \in \calA_\mu^p,
\end{align*}
where the dithers $\tauv_k^\triangle$ are independent and drawn from the \emph{triangular distribution} 
\begin{equation}
	\label{eq:tridither}
	\calU_\R(-\mu,\mu)^p * \calU_\R(-\mu,\mu)^p,
\end{equation}
which is the convolution of two uniform distributions, see Figure \ref{fig:Dither}. {To cleanly distinguish the settings, we highlight related quantities with a $\triangle$-symbol in the superscript.}
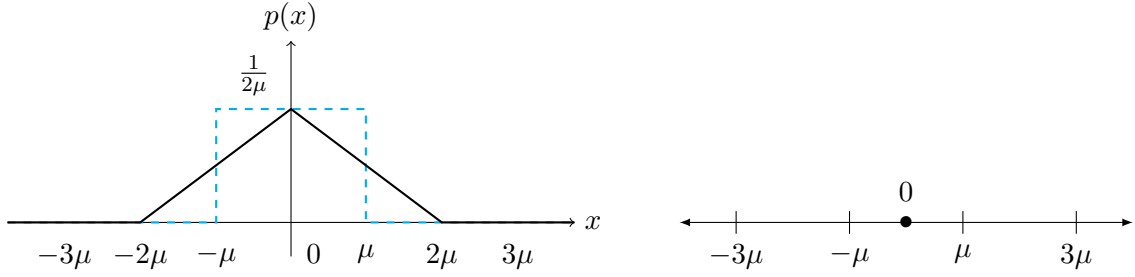
\begin{figure}[t]
	\centering
	\begin{tikzpicture}[scale=1.5]
		\draw[->] (-2.5,0) -- (2.5,0) node[right] {$x$};
		\draw[->] (0,-0.3) -- (0,1.6) node[above] {$p(x)$};
		
		\draw[thick, cyan, dashed] (-2.5,0) -- (-0.66,0) -- (-0.66,1) -- (0.66,1) -- (0.66,0) -- (2.5,0);
		
		\draw[thick, black] (-2.5,0) -- (-1.33,0) -- (0,1) -- (1.33,0) -- (2.5,0);
		
        \node at (-2,-0.1) [below] {$-3\mu$};
        \node at (2,-0.1) [below] {$3\mu$};
		\node at (-1.33,-0.1) [below] {$-2\mu$};
		\node at (1.33,-0.1) [below] {$2\mu$};
		\node at (-0.66,-0.1) [below] {$-\mu$};
		\node at (0.66,-0.1) [below] {$\mu$};
		\node at (0.2,-0.1) [below] {$0$};
		\node at (-0.1,1.3) [left] {$\tfrac{1}{2\mu}$};
	\end{tikzpicture}
	\qquad 
	\begin{tikzpicture}[scale=1.5]
		\draw[latex-latex] (-2,0) -- (2,0);
		
		\draw (-1.5,0.1) -- (-1.5,-0.1) node[below] {$-3\mu$};
		\draw (-0.5,0.1) -- (-0.5,-0.1) node[below] {$-\mu$};
		\node at (0,0) {\textbullet};
		\node at (0,0.1) [above] {$0$};
		\draw (0.5,0.1) -- (0.5,-0.1) node[below] {$\mu$};
		\draw (1.5,0.1) -- (1.5,-0.1) node[below] {$3\mu$}; 
	\end{tikzpicture}
    
	\caption{Densities of triangular vs uniform distribution (left), and alphabet $\calA_\mu$ (right).}
	\label{fig:Dither}
\end{figure}

\begin{remark}
\label{rem:Bbit}
    {For bounded or strongly concentrating data and noise distributions as in Assumption \ref{assump:main}, $Q_\mu$ effectively behaves like the sign-quantizer $Q(\cdot) = \sign(\cdot)$ if $\mu$ is chosen sufficiently large. Hence for smaller $\mu > 0$, $Q_\mu$ is a natural generalization of $Q$ to (finite) multi-bit regimes and the infinite-range alphabet just simplifies theoretical analysis.
    Let us make this more precise.} 
    \begin{enumerate}
        \item[(i)] For bounded inputs, triangular dithering uses a finite number of bits, depending on the choice of resolution $\mu$. We first see that $Q_\mu$ maps $\left[ 2k\mu, (2k+2)\mu \right)$ to the point $(2k+1)\mu$, for any $k \in \Z$.
        If $\lambda>0$ is such  that $|x|\leq \lambda$, then $|x+\tau^\triangle|<\lambda+2\mu$ a.s., so picking $\mu=\lambda/(2^b-2)$ for any integer $b\geq 2$, we see that $Q_\mu$ maps
        $(-\lambda-2\mu,\lambda+2\mu)$ to $$\calA_\mu \cap [-(2^b-2)\mu -\mu, (2^b-2)\mu + \mu].$$
        This consists of $2^b$ elements, so this choice of $\mu$ yields a triangular quantizer that uses $b$ bits. Note that $b=2$ is the smallest allowable choice of $b$ for the triangular dither.
        \item[(ii)] {In the case of unbounded distributions that strongly concentrate, the same argumentation holds with high probability. Replacing $Q_\mu$ in practice with a finite bit quantizer that saturates outside of $[-\lambda-2\mu,\lambda+2\mu]$ yields a feasible method for which all derived guarantees stay valid.} From here on, if $\lambda$ is any upper bound for $|x|$ {that holds w.h.p.}, we refer to $Q_\mu$ with $\mu=\lambda/(2^b-2)$ as {\it $b$-bit triangular dithering (or quantizer)}.
    \end{enumerate}
\end{remark}

The statistical properties of such dithered quantizers were carefully studied, and we refer the reader to Gray's survey \cite{gray1993dithered} and the references therein. We make one important change in perspective from Gray. Even though $\tauv_k^\triangle$ is a dither and $\by_k$ is the signal, it is actually more meaningful to treat $\be_k+\tauv_k^\triangle$ as the {\it effective dither} and $\bx_k$ as the signal. Accordingly, we define the {\it quantization noise} as
\begin{equation}
	\label{eq:quannoise}
	\bxi_k 
	:= \bq_k^\triangle-\bx_k
	=Q_\mu(\bx_k + \be_k + \tauv_k^\triangle) - \bx_k.
\end{equation}
In general, we cannot say that $\bxi_k$ is independent of $\bx_k$. However, due to the triangular distribution of the dithering vectors, the random variables $\bxi_1,\dots,\bxi_n$ have several favorable properties as long as the distribution of $\be$ satisfies the following mild regularity assumption. 

\begin{assumption}\label{assump:noise2}
	Assume that {the coordinates} of the distribution of $\be$ {are uncorrelated,} absolutely continuous with respect to the Lebesgue measure on $\F$, and that the characteristic function of each coordinate of $\be$ is twice differentiable. 
\end{assumption}
\begin{lemma}
	\label{lem:triquannoise}
	There is an absolute constant $C>0$ such that the following holds. Let $\bx_1,\ldots,\bx_n \in \R^p$. Let $\be_1,\ldots,\be_n\overset{\textnormal{i.i.d.}}{\sim}\be$. Suppose that $\be$ is $K$-subgaussian for some $K\geq 1$ and satisfies Assumption \ref{assump:noise2} with mean zero and entry-wise variance $\nu^2$. Then $\bxi_1,\dots,\bxi_n$ are independent, $CK$-subgaussian with mean zero and covariance $(\mu^2+\nu^2)\id_p$.
\end{lemma}

The proof of Lemma \ref{lem:triquannoise} is provided in Section \ref{sec:triquannoise}.

\begin{remark}
	Let us highlight two points before we proceed:
	\begin{enumerate}
		\item[(i)] Lemma \ref{lem:triquannoise} extends to the complex setting. For a complex valued input $\by_k$, we define $\bq_k^\triangle\in \calA_\mu^p + \i \calA_\mu^p$ as
		$$
		\bq_k^\triangle:=Q_{\mu}(\Re(\by_k)+\tauv_k^\triangle)+\i Q_{\mu}(\Im(\by_k)+ \dot \tauv_k^\triangle),
		$$
		where $\tauv_k^\triangle$ and $\dot \tauv_k^\triangle$ are independent and drawn from the same triangular distribution \eqref{eq:tridither}. Since we assume that $\Re(\be_k)$ and $\Im(\be_k)$ are independent, $\Re(\bq_k)$ and $\Im(\bq_k)$ are independent and the conclusions in \cref{lem:triquannoise} hold for both parts.
        \item[(ii)] In the special case that $\be=\calU_{\R}(-\mu,\mu)^p$, one can set $\tauv\sim \calU_{\R}(-\mu,\mu)^p$ instead of \eqref{eq:tridither}, since then, the effective dither $\be+\tauv$ is triangular. In which case, the conclusions of \cref{lem:triquannoise} hold with $\bxi_k\mid \bx_k$ having $\psi_2$ norm at most $C \mu$ and $\E \xi_{k,\ell}^2= \mu^2$ for each $k\in [n]$ and $\ell\in [p]$.
	\end{enumerate}
\end{remark}

Let us now return to our initial goal of estimating the leading eigenspace of $\SigmaX$ from the quantized observations $\bq_1^\triangle, \dots, \bq_n^\triangle$.
Rearranging \eqref{eq:quannoise}, we get $\bq_k^\triangle=\bx_k+\bxi_k$ for each $k\in [n]$, which can be written in matrix form as  
\begin{equation} \label{eq:Qtri}
	\bQ_n^\triangle 
	= \bX_n + \bXi_n. 
\end{equation}
Lemma \ref{lem:triquannoise} states that given $\bx_k$ the entries of $\bxi_k$ are uncorrelated, mean-zero, and have variance $\sigma^2:=c_\F^2(\nu^2 + \mu^2)$ where $c_\R = 1$ and $c_\C=2$. This implies that
\begin{align*}
	\E(\bq_k^\triangle (\bq_k^\triangle)^*)
	&=\E( \bx_k \bx_k^*+ \bx_k \bxi_k^* + \bxi_k \bx_k^* + \bxi_k \bxi_k^*) \\
	&= \E \bx_k \bx_k^* + \E \bxi_k \bxi_k^* \\
	&= \SigmaX+ \sigma^2 \id_p.
\end{align*}
{Since the shift by $\sigma^2\id_p$ leaves the eigenspaces of $\SigmaX$ invariant}, it is natural to define the triangular estimator as
\begin{equation}
		\label{eq:Sigmatri}
		\SIGMADITHtriangular_n:= \frac 1 n \sum_{k=1}^n \bq_k^\triangle (\bq_k^\triangle)^*
\end{equation}
We deduce the following result for estimating the leading eigenspace of $\SigmaX$ via $\SIGMADITHtriangular_n$. 

\begin{theorem} \label{thm:subspacerectangletriangle}
	Suppose Assumptions \ref{assump:main} and \ref{assump:noise2} hold and that $\rank(\SigmaX) = s$. There are absolute constants {$C_1,C_2 > 0$}, and constants $c_K,\alpha_K>0$ depending only on $K$ such that the following holds. Let $\bU,\bU_n^\triangle \in \O^{p\times s}$ denote orthonormal bases for the leading left singular spaces of $\SigmaX$ and $\SIGMADITHtriangular_n$.  
		\begin{enumerate}
			\item[(i)] For deterministic $\bx_1,\dots,\bx_n \subset \F^p$ and any $\alpha \geq \alpha_K$, we have with probability at least $1-e^{-c_K \alpha p}$
			\begin{equation}
				\label{eq:CZvariationineq}
				\dist(\bU_n^\triangle,\bU)
				\leq \min \left\{ 1, \, {C_1} \sqrt{\frac {\nu^2+\mu^2}{\lambda_s(\SigmaX)}} \left(1 +\sqrt{ \frac{\nu^2+\mu^2}{\lambda_s(\SigmaX)} } \right) \sqrt{\frac{\alpha p}{n}} \right\}. 
			\end{equation}
			\item[(ii)] For $\bx_1,\dots,\bx_n \overset{i.i.d.}{\sim} \bx$ with $\bx \sim \mathcal \calC\calN(\0,\SigmaX)$ and any $\alpha \ge \alpha_K$, 
            {we have that if
            $n \ge C_2 \frac{\sum_{j=1}^s \lambda_j(\SigmaX)}{\lambda_s(\SigmaX)}$
            then with probability at least $1-3e^{-c_K \min\{\alpha p,n\} }$}
			\begin{equation}
				\label{eq:CZvariationineqCorollary}
				\dist(\bU_n^\triangle,\bU)
				\leq \min \left\{ 1, \, {C_1} \sqrt{\frac {\nu^2+\mu^2}{\lambda_s(\SigmaX)} } \left( 1 + \sqrt{ \frac{\nu^2+\mu^2}{\lambda_s(\SigmaX)} } \right) \sqrt{\frac{\alpha p}{n}} \right\}.
			\end{equation}
	\end{enumerate}
\end{theorem}
Theorem~\ref{thm:subspacerectangletriangle}(i) is a variation of \cite[Theorem 3]{cai2018rate}. Their result cannot be directly used for our purposes because it relies on the stronger assumption that the noise ($\bXi_n$ in this case) has i.i.d.\ sub-Gaussian entries. Our proof of \cref{thm:subspacerectangletriangle}(i) follows the same strategy as in \cite{cai2018rate} with suitable modifications, and can be found in \cref{sec:CZvariation}. \cref{thm:subspacerectangletriangle}(ii) then follows from (i) by a high probability bound of $\lambda_s(\tfrac{1}{n} \bX_n\bX_n^*)$ in terms of $\lambda_s(\SigmaX)$ that can be found in Section~\ref{sec:CZvariationII}.

\begin{remark}
    We restrict ourselves to complex Gaussian distributions here for the sake of simplicity and their widespread use in theoretical works on spectral estimation. Note, however, that Theorem \ref{thm:subspacerectangletriangle}(ii) can be generalized to heavier-tailed distributions, cf.\ Theorem \ref{thm:subspacerectangletriangle_heavytailed}.
\end{remark}

\subsection{Comparisons between rectangular and triangular dithering}
\label{sec:comparison}

Let us carefully compare \cref{thm:subspacerectangle} for rectangular and \cref{thm:subspacerectangletriangle} for triangular {dithering}. 
On the one hand, \cref{thm:subspacerectangle} puts milder assumptions on the eigengap of $\SigmaX$  
{and does not require a rank-defective covariance matrix. On the other hand,} the error bound in \cref{thm:subspacerectangletriangle} improves on the dependence of the condition number of $\SigmaX$ for low-noise regimes with sufficiently high bit rates, i.e., $\nu^2+\mu^2\lesssim \lambda_r(\SigmaX)$.

To make this more precise, let us consider noisy samples $\by_1,\dots,\by_n$ following Assumptions \ref{assump:main} and \ref{assump:noise2} with underlying deterministic data $\bx_1,\dots,\bx_n$ such that $\SigmaX = \frac{1}{n} \bX_n \bX_n^*$ has rank $r$.\footnote{For random data, the comparison follows analogously.} Using \cref{thm:subspacerectangle} with $s=r$, $\lambda^2=C_K C_{\infty}\log(n)$ for a $C_K$ that only depends on $K$, and $t=\alpha \log(p)$, for any $\alpha\geq 1$ and $n\geq (\alpha+1)\lambda^2 p \log(p)$,  we have  with probability at least $1-p^{-\alpha}$ that 
\begin{equation}
	\label{eq:EVbound}
	\dist(\bU_n^\square,\bU)
	\leq \min \left\{ 1, \, \frac{C_K(\Cop^{1/2}+\lambda)}{\lambda_r(\SigmaX)} \sqrt{\frac{(\alpha+1) p \log (p)}{n}} \right\}.
\end{equation}
To see the dependence of the right side on $\SigmaX$, first notice that
$$
\Cop 
= \max_{k\in [n]} \|\bx_k\bx_k^*\|^2+\nu^2 
\geq \lambda_1(\SigmaX)+\nu^2. 
$$
Defining {$\kappa_r(\SigmaX) := \lambda_1(\SigmaX)/\lambda_r(\SigmaX)$ as the $r$-th condition number} of the covariance matrix, the right hand side of \eqref{eq:EVbound} decays not faster than
\begin{align}
    \label{eq:OmegaBound}
    \Omega \left( \max \left\{ \sqrt{\frac{{\kappa_r(\SigmaX)}}{\lambda_r(\SigmaX)}}, \, {\sqrt{\frac{\nu^2 +\lambda^2}{\lambda_r(\SigmaX)}}} \right\} \sqrt{\frac{{\alpha} p}{n}}\right).
\end{align}

\cref{thm:subspacerectangletriangle} provides a comparable rate in both $n$ and $\lambda_r(\SigmaX)$, unless $\nu^2+\mu^2 \lesssim \lambda_r(\SigmaX)$. 
For a sufficiently large absolute constant $c>0$, let 
\begin{align}
	\label{eq:BitRate}
	b \geq \left\lceil \log_2\left( 2+\frac{c}{\sqrt{\lambda_r(\SigmaX)}} \max_{k\in[n]} \| \by_k \|_\infty \right) \right\rceil. 
\end{align}
If $\nu^2 \lesssim \lambda_r(\SigmaX)$ and $b$ satisfies \eqref{eq:BitRate}, then by Theorem \ref{thm:subspacerectangletriangle} there is an absolute $C>0$ such that, with probability at least $1-e^{-c_K \alpha p}$,
\begin{equation}
	\label{eq:SVDbound}
	\dist(\bU_n^\triangle,\bU)
	\leq \min \left\{ 1, \, C\sqrt{\frac {\nu^2+\mu^2}{\lambda_r(\SigmaX)}} \sqrt{\frac{\alpha p}{n}} \right\}. 
\end{equation}
{Comparing \eqref{eq:OmegaBound} with \eqref{eq:SVDbound} for ill-conditioned problems, the latter bound clearly improves on \eqref{eq:EVbound} whenever the additional assumptions hold.}

\section{Performance guarantees for ESPRIT in quantized DOA}
\label{sec:ESPRITforDOA}

We now apply the {quantized subspace estimation techniques} to multi-snapshot spectral estimation {also known as DOA estimation}. A mathematical formulation of the problem without quantization is as follows. 

\begin{assumption}
    \label{assump:specestimationproblem}
    Let $s,p,n\in \N$ with $n>p > s$. Suppose $\btheta=(\theta_1,\dots,\theta_s)\in (\R/\Z)^s=[0,1)^s$ consists of distinct elements. 
    \begin{enumerate}[(i)]
        \item 
        Let $\ba_1,\dots,\ba_n\in \F^s$ be deterministic (d) or stochastic (s) vectors. We then either
        \begin{enumerate}
            \item[(d)] assume $\SigmaA:=\frac 1 n \sum_{k=1}^n \ba_k \ba_k^*$ has rank $s$ or 
            \item[(s)] assume $\ba_1,\dots,\ba_n\overset{\text{i.i.d.}}{\sim} \ba$, $\ba\in \F^p$ is $K/2$-subgaussian, and $\SigmaA:=\E \ba \ba^*$ has rank $s$. 
        \end{enumerate}
        \item 
        Let $\be_1,\dots,\be_n\overset{\text{i.i.d.}}{\sim} \be$, where $\be\in\F^p$ is $K$-subgaussian with uncorrelated entries that are mean zero and have variance $\nu^2$. If $\F=\C$, then we assume that the real and imaginary parts of the entries of $\be$ are independent.
        \item 
        Let $\ba_1,\dots,\ba_n,\be_1,\dots,\be_n$ be independent. 
        \item 
        For each $k \in [n]$, set
        \begin{align*}
        \bPhi&:=\bPhi(\btheta) = \Big[ e^{-2\pi \i j \theta_\ell} \Big]_{j\in [p]-1,\, \ell \in [s]} \in {\C^{p \times s}}, \\
        \by_k&:=\bPhi \ba_k+\be_k \in \C^{p}.
        \end{align*}
    \end{enumerate}
\end{assumption}

Here, $\btheta=(\theta_1,\ldots,\theta_s)$ represents a collection of $s$ point sources (or frequencies or objects) at distinct locations (or angles) and $\ba_1,\dots,\ba_n\in \F^s$ represent their amplitudes at various times, say $t_1,\dots,t_n$. The amplitudes can be deterministic or stochastic.  
As before, $\be_k$ represents the noise. The noiseless data $\bPhi \ba_k$ {consists of} the first $p$ Fourier coefficients of the point sources at time $t_k$. The goal is to estimate $\btheta$ as accurately as possible from just the data $\bY_n$ and perhaps additional weak information. This inverse problem, {which goes by many names such as spectral estimation, super-resolution, and direction-of-arrival estimation,} is nonlinear since $\btheta$ is unknown and $\bPhi$ generally has non-orthogonal columns. 

We are interested in {its quantized counterpart}, whereby each analog measurement $\by_k$ is put through a quantizer $Q$ of our design, to output $\bq_k := Q(\by_k).$ The introduction of $Q$ models hardware limitations. One typically needs $p$ sensors to collect $p$ Fourier coefficients, and since the production of a high-resolution sensor is expensive, collecting a large number of high-precision Fourier measurements may be infeasible. Instead, it may be better to use a greater number of low precision sensors, but utilize them in a smarter way, such that accurate estimation of $\btheta$ is {still }possible.

Spectral estimation is intrinsically linked to subspace estimation. This stems from the fact that estimating the ``signal subspace" on which $\bx_1,\dots,\bx_n$ lie is an important first step in celebrated spectral estimation algorithms like MUSIC \cite{schmidt1986multiple}, ESPRIT \cite{kailath1989esprit}, and MPM \cite{hua1990matrixpencil}. Our results on subspace estimation are relevant to all three algorithms, but a complete treatment is beyond the scope of this paper. We will focus on (multi-snapshot) ESPRIT since it is a computationally efficient algorithm with provable recovery guarantees. 

An accurate deterministic perturbation and resolution analysis of ESPRIT has been worked out in \cite{li2022stability}. Postponing the technical details for the time being, the authors showed that if $\hat\bU$ is a sufficiently good approximation of $\bU$, an orthonormal basis for the range of $\bPhi$, then the output $\widehat\btheta$ of ESPRIT satisfies
\begin{equation}
	\label{eq:mdsintheta}
	\md(\hat\btheta,\btheta) \lesssim_{s} \dist(\hat\bU,\bU). 
\end{equation}
Here, we define $|\theta|_\T:=\min_{n\in\Z} |\theta+n|$ and $\md$ denotes the matching distance,   
$$
\md(\hat\btheta,\btheta)
:= \min_{{\pi}} \max_k |\theta_k-\hat\theta_{{\pi(k)}}|_\T,
$$
where the minimum is taken over all permutations {$\pi$} of $[s]$. A formal version of \eqref{eq:mdsintheta} is stated in \cref{lem:mdsintheta}.

Inequality \eqref{eq:mdsintheta} {implies} that any good estimator $\hat\bU$ for $\bU$ can be used in the ESPRIT algorithm to obtain recovery guarantees for spectral estimation. In particular, we generate subspace estimates $\bU_n^\square$ or $\bU_n^\triangle$ from quantized samples, see \cref{alg:subspace}, to obtain estimators $\btheta_n^\square$ and $\btheta_n^\triangle$ of $\btheta$ via ESPRIT, cf.\ \cref{alg:espritquan}. 

\begin{algorithm}[t]
	\begin{algorithmic}
		\Require Parameter $s$ and dithered quantized measurements $\{\bq_k^\square, \dot \bq_k^\square\}_{k=1}^{n}$ or $\{\bq_k^\triangle\}_{k=1}^{n}$.
		\State 1. Form the covariance estimator $\SIGMADITHuniform_n$ {in \eqref{eq:TwoBitEstimator}} or $\SIGMADITHtriangular_n$ {in \eqref{eq:Sigmatri}}. \smallskip 
		\State 2. Compute the leading $s$ eigenvectors of $\SIGMADITHuniform_n$ or $\SIGMADITHtriangular_n$. \smallskip  
		\Ensure  $\bU_n^\square$ or $\bU_n^\triangle$.
	\end{algorithmic}
	\caption{Subspace estimation with quantized measurements}
	\label{alg:subspace}	
\end{algorithm}

\begin{algorithm}[t]
	\begin{algorithmic}
		\Require Number of sources $s$, quantized measurements either $\{\bq_k^\square, \dot \bq_k^\square\}_{k=1}^{n}$ or $\{\bq_k^\triangle\}_{k=1}^n$. \medskip 			
		\State 1. Let $\hat\bU\in \O^{p\times s}$ be either $\bU_n^\square$ or $\bU_n^\triangle$, which are computed according to \cref{alg:subspace}. \smallskip 
		\State 2. Let $\hat\bU_{0}$ and $\hat\bU_{1}$ be two submatrices of $\hat\bU$ containing the first and last $p-1$ rows, respectively. Compute the $s$ eigenvalues $\hat\lambda_1,\dots,\hat\lambda_s$ of the matrix $\hat \bPsi:=\hat\bU_{0}^{\dagger} \hat\bU_{1}.$ 
		\medskip 
		\Ensure $\hat\btheta=(\hat\theta_1,\dots,\hat\theta_s)$, where $\hat\theta_k=-\arg(\hat\lambda_k)/(2\pi)$ and $\arg(re^{\i \varphi})=\varphi\in [0,2\pi)$. 
	\end{algorithmic}
	\caption{Multi-snapshot ESPRIT with quantized measurements}
	\label{alg:espritquan}	
\end{algorithm}

From the analysis of ESPRIT in the unquantized setting \cite{li2022stability}, we know that its stability to noise greatly depends on the geometry of $\btheta$. Recall that the minimum separation of $\btheta$ is defined as 
$$
\Delta
:=\Delta(\btheta)
:=\min_{k\not=k'} |\theta_k-\theta_{k'}|_{\T}.
$$
The best possible scenario is when $\Delta(\btheta)$ is large compared to $\frac 1p$, where $p$ stands for the number of Fourier measurements. This is evident in the condition number {$\kappa_s(\bPhi)$} of $\bPhi$. It was shown in \cite{aubel2017vandermonde} that $\Delta(\btheta)> \frac 1 p$ implies
\begin{equation}
	\label{eq:wellsep}
	p - \frac 1{\Delta(\btheta)}
	\leq  \sigma_s^2(\bPhi)
	\leq \sigma_1^2(\bPhi)
	\leq p + \frac 1{\Delta(\btheta)}.
\end{equation}
This is of course only possible when enough Fourier samples are collected. The opposite situation is also of interest, especially in the context of super-resolution. When $\Delta(\btheta)$ is significantly smaller than $\frac 1 p$, the condition number of $\bPhi$ no longer admits a lower bound with a simple expression. It predominantly depends on the number of elements in $\btheta$ that lie within an interval on the order of $\frac 1 p$ and their geometry at small scales. We refer the reader to \cite{li2024multiscale} for further details and the state-of-the-art results.

\subsection{Four-bit spectral estimation by rectangular dithering}

We are ready to derive our main results for ESPRIT when the measurements are quantized via a rectangular dither as outlined in \cref{sec:rectangular}. Recall that each analog measurement $\by_k\in \C^p$ is quantized to $\bq_k^\square,\dot\bq_k^\square\in \{\pm 1 \pm \i\}^p$. {Despite being cheap to implement via one-bit ADCs, see Remark \ref{rem:ADC}}, in total this is a four-bit quantizer since each $(\by_k)_j\in \C$ is quantized to a pair of complex numbers $((\bq_k)_j, (\dot \bq_k)_j)\in \{\pm 1 \pm \i\}^2$, which can be specified by four bits. Ultimately, it thus requires $4n$ bits to store all $n$ quantized samples.

We start with the general setting where no assumptions on $\Delta(\btheta)$ are made. The following theorem is proved in \cref{proof:espritrectangular}. {The factors $\gamma_1$ and $\gamma_2$ appearing therein are defined depending on whether the amplitudes are deterministic or stochastic, see Assumption \ref{assump:specestimationproblem}. 
\begin{enumerate}[(i)]
		\item 
        For deterministic amplitudes, let $\gamma_1:=\max_{k\in [n]} \|\ba_k\ba_k^*\|_1$ and $\gamma_2:=\max_{k\in [n]} \|\ba_k\|_2$. 
        \item 
        For stochastic amplitudes, let $\gamma_1:=\|\SigmaA\|_1$ and $\gamma_2:=\sqrt{\lambda_1(\SigmaA)}$.
\end{enumerate}
}

\begin{theorem}
	\label{thm:espritrectangular}
    Suppose \cref{assump:specestimationproblem} holds. There exist constants $B_K,C_K>0$ that depend only on $K$ such that the following hold. 
    For any $\lambda^2 \geq B_K \log(n) (\gamma_1+ \nu^2)$ and $t\geq 1$, whenever 
    \begin{equation}
			\label{eq:ncondition}
			n
            \geq \max\left\{ (t+1)\lambda^2 p \log(p), \, \frac{ C_K^2 16^s s^3 (\gamma_2\sigma_1(\bPhi)+\nu+\lambda)^2 (t+1)}{\lambda_s^2(\SigmaA)} \frac{p^3 \log(p)}{\sigma_s^8(\bPhi) \Delta^2} \right\},
	\end{equation}
	with probability at least $1-p^{-t}$, we have	{for $\btheta^\square_n$ computed via Algorithms \ref{alg:subspace} and \ref{alg:espritquan} that}
	\begin{equation} \label{eq:espritrect2}
		\md(\btheta^\square_n,\btheta)
		\leq \min \left\{ \frac 12,  \,  \frac{C_K 4^s (\gamma_2\sigma_1(\bPhi)+\nu+\lambda)}{\sigma_s^2(\bPhi) \lambda_s(\SigmaA)}
           \sqrt{\frac{(1+t)p \log (p)}{n}} \right\}.
	\end{equation}
\end{theorem}

We conclude this section by considering the well-separated case where $\Delta(\btheta)\geq \frac 2 p$. In this case, we use inequality \eqref{eq:wellsep} and hence we obtain the following corollary.

\begin{corollary}
	\label{cor:espritrectangular}
    Suppose \cref{assump:specestimationproblem} holds and that $\Delta(\btheta)\geq \frac 2 p$. There exist constants $B_K,C_K>0$ that depend only on $K$ such that the following hold. 
	For any $\lambda^2 \geq B_K \log(n) (\gamma_1+ \nu^2)$ and $t\geq 1$, whenever 
    \begin{equation*}
        n\geq \max\left\{(t+1)\lambda^2 p \log(p), \, \frac{ C_K^2 16^s s^3 (\gamma_2 \sqrt p+\nu+\lambda)^2 (t+1) p \log(p)}{\lambda_s^2(\SigmaA)}\right\} ,
	\end{equation*}
	with probability at least $1-p^{-t}$, we have	{for $\btheta^\square_n$ computed via Algorithms \ref{alg:subspace} and \ref{alg:espritquan} that}
	\begin{equation*}
	    \md(\btheta^\square_n,\btheta)
	    \leq \min \left\{ \frac 12,  \,  \frac{C_K 4^s (\gamma_2 \sqrt p +\nu+\lambda)}{\lambda_s(\SigmaA) \sqrt p}      \sqrt{\frac{(1+t)\log (p)}{n}} \right\}.
	\end{equation*}
\end{corollary}

\subsection{Multi-bit spectral estimation by triangular dithering}

{Let us now turn to analyzing} ESPRIT when the measurements are quantized via a triangular dither as outlined in \cref{sec:triangular}. Recall that each analog measurement $\by_k\in \C^p$ is quantized to $\bq_k^\triangle \in \calA_\mu$. {Further recall from Remark \ref{rem:Bbit} that while they are proved for the infinite range quantizer $Q_\mu$, the results also apply to finite $b$-bit quantizers by concentration of the data and noise distributions.}
We start with a general result where $\Delta(\btheta)$ is arbitrary. The following theorem  is proved in \cref{proof:esprittriangular}. 

\begin{theorem}
	\label{thm:esprittriangular}
	Suppose Assumptions \ref{assump:noise2} and \ref{assump:specestimationproblem} hold and assume that 
    \begin{equation}
		\label{eq:noisecond1}
		\nu^2+\mu^2\leq \sigma_s^2(\bPhi)\lambda_s(\SigmaA).
	\end{equation}
    There exist absolute constants $C,c_1>0$, and constants $c_K,\alpha_K>0$ depending only on $K$, such that the following holds. Let {$\btheta^\triangle_n$ be computed via Algorithms \ref{alg:subspace} and \ref{alg:espritquan}.}
	\begin{enumerate}[(i)]
		\item 
		For deterministic $\ba_1,\dots,\ba_n$ and any $\alpha\geq \alpha_K$, assume that 
		\begin{equation}
			\label{eq:ncond1}
			n \geq \frac{C^2 16^s s^4 (\nu^2+\mu^2)}{\lambda_s(\SigmaA)}  \frac {\alpha p^2}{\sigma_s^6(\bPhi)\Delta^2}.
		\end{equation}
        Then, with probability at least $1-e^{-c\alpha p}$,
    	\begin{equation}
    		\label{eqn:esprittriangularEst}
    		\md(\btheta_n^\triangle,\btheta) 
    		\leq \min \left\{ \frac 12, \, \, C 4^s \sqrt{\frac {\nu^2+\mu^2}{\sigma_s^2(\bPhi)\lambda_s(\SigmaA)}}\sqrt{\frac{\alpha p}{n}}  \right\}.
    	\end{equation}
		\item 
        For $\ba_1,\dots,\ba_n \overset{i.i.d.}{\sim} \ba$ with $\ba \sim \mathcal \calC\calN(\0,\SigmaA)$, any $\alpha \ge \alpha_K$, and any $\beta \ge 1$, assume that 
        \begin{equation}
			\label{eq:ncond2}
			n \geq \max\left\{ {c_1 \beta \Bigg(\sum_{j=1}^s \kappa_j^4(\bPhi)\Bigg)^{1/2} \Bigg(\sum_{j=1}^s \kappa_j^2(\SigmaA)\Bigg)^{1/2}}, \, \frac{C^2 16^s s^4 (\nu^2+\mu^2)}{\lambda_s(\SigmaA)}  \frac{\alpha p^2}{\sigma_s^6(\bPhi)\Delta^2} \right\}.
		\end{equation}
        Then, with probability at least $1-3e^{-c_K \min\{\alpha p, \beta r \}}$, 
 		\begin{equation}
    		\label{eqn:esprittriangularEst2}
    		\md(\btheta_n^\triangle,\btheta) 
    		\leq \min \left\{ \frac 12, \, \, C 4^s \sqrt{\frac {\nu^2+\mu^2}{\sigma_s^2(\bPhi)\lambda_s(\SigmaA)}}\sqrt{\frac{\alpha p}{n}}  \right\}.
    	\end{equation}
	\end{enumerate}
\end{theorem}

{The noise assumption \eqref{eq:noisecond1} is not restrictive because the terms on both sides are independent of $n$, and $\sigma_s(\bPhi)$ alone grows in $p$.} Hence, the assumption will always be satisfied for fixed $\btheta$ and sufficiently large $p$. Moreover, recall that $\sigma_1(\bPhi)\in [\sqrt p, \sqrt{ps}]$. Hence, 
$
p\sigma_s^{-2}(\bPhi)
\leq {\kappa_s^2(\bPhi)}
$
so that the above estimates can be reformulated in terms of {$\kappa_s(\bPhi)$}. 

While the conditions and expressions in \cref{thm:espritrectangular} look complicated, they simplify into natural expressions in the well-separated case where $\Delta(\btheta)\geq \frac 2 p$. In this case, we use \eqref{eq:wellsep} to obtain the following corollary.

\begin{corollary}
	\label{cor:esprittriangular}
	Suppose Assumptions \ref{assump:noise2} and \ref{assump:specestimationproblem} hold, and assume that 
    $$
	\Delta(\btheta)\geq \frac 2 p \andspace 4(\nu^2+\mu^2)\leq p\lambda_s(\SigmaA).
	$$
    There exist absolute constants $C,c_1>0$, and constants $c_K,\alpha_K>0$ depending only on $K$, such that the following holds. Let {$\btheta^\triangle_n$ be computed via Algorithms \ref{alg:subspace} and \ref{alg:espritquan}.}
	\begin{enumerate}[(i)]
		\item 
		For deterministic $\ba_1,\dots,\ba_n$ and any $\alpha\geq \alpha_0$, assume that 
		\begin{equation*}
			n \geq \frac{C^216^s s^4 (\nu^2+\mu^2) \alpha p}{\lambda_s(\SigmaA)}.
		\end{equation*}
        Then, with probability at least $1-e^{-c_K\alpha p}$,
    	\begin{equation*}
    		\md(\btheta_n^\triangle,\btheta) 
    		\leq \min \left\{ \frac 12, \, C 4^s \sqrt{\frac{\nu^2+\mu^2}{\lambda_s(\SigmaA)}}\sqrt{\frac{\alpha }{n}} \right\}.
    	\end{equation*}
		\item 
        For $\ba_1,\dots,\ba_n \overset{i.i.d.}{\sim} \ba$ with $\ba \sim \mathcal \calC\calN(\0,\SigmaA)$, any $\alpha \ge \alpha_K$, and any $\beta \ge 1$, assume that 
        \begin{equation*}
            n \geq \max\left\{ {c_1 \beta \sqrt{s} \Bigg(\sum_{j=1}^s \kappa_j^2(\SigmaA)\Bigg)^{1/2}}, \, \frac{C^2 16^s s^4 (\nu^2+\mu^2) \alpha p }{\lambda_s(\SigmaA)}  \right\}.
		\end{equation*}
        Then, with probability at least $1-3e^{-c_K \min\{\alpha p, \beta r \}}$,
		\begin{equation*}
    		\md(\btheta_n^\triangle,\btheta) \le 
    		\min \left\{ \frac 12, \, C 4^s \sqrt{\frac{\nu^2+\mu^2}{\lambda_s(\SigmaA)}}\sqrt{\frac{\alpha }{n}}  \right\}.
    	\end{equation*}
	\end{enumerate}
\end{corollary}

\begin{remark}
    The error bounds derived in this section primarily focus on dependence on $n$, $\nu$, and $\lambda_s(\bSigma_{\ba})$, and treats all other quantities as constants. This section contains inequalities that have an exponential dependence on $s$, which come from \cite[Theorem 3]{Li20}. Any improvements on the latter result would directly improve the bounds in Theorems \ref{thm:espritrectangular} and \ref{thm:esprittriangular}. There is a different result \cite[Theorem 4]{Li20} that only applies to the case where $\Delta(\btheta)\geq C/p$, which reduces exponential to polynomial dependence on $s$. By using this more specific result and mimicking the proofs of \cref{thm:espritrectangular,thm:esprittriangular}, one can show that \cref{cor:espritrectangular,cor:esprittriangular} hold with polynomial growth in $s$. 
\end{remark}

\section{Numerical experiments}
\label{sec:Numerics}

In this section, we will compare the performance {of our estimators for} rectangular and triangular dithering on various subspace and spectral estimation tasks. 
{Recall from \cref{sec:rectangular} and Remark \ref{rem:Bbit} that our theory for both estimators are applicable to finite-range quantizers provided that the data concentrates and appropriate parameters $\lambda$ resp. $\mu$ are used.}

{In part of the experiments, we will use {\it $b$-bit direct rounding} as a benchmark}. The latter is a map $R_{\lambda,b}$ defined on $[-\lambda,\lambda]$ such that 
$$
R_{\lambda,b}(x):= 
\begin{cases}
	Q_{\lambda/2^b}(x) &\text{if } x\in [-\lambda,\lambda), \\
	\lambda -  \frac{\lambda}{2^b} &\text{if } x=\lambda.
\end{cases}
$$
{Recall the definition of $Q_{\lambda/2^b}$ defined in \eqref{eq:Qmudef}.} Note that $
|R_{\lambda,b}(x)-x|\leq \frac \lambda {2^b}$ whenever $|x|\leq \lambda$.
{The $b$-bit direct rounding extends to complex vectors by applying it to the real and imaginary parts of each entry separately. Hence, all of the compared methods store} $c_\F b p$ many bits per vector in $\F^p$, where we recall that $c_\R=1$ and $c_\C=2$. Furthermore, if $n$ vectors in $\F^p$ are quantized, then this requires $c_\F b p n$ bits. 
{Since in the literature $b$-bit direct rounding is often simply referred to as MSQ (memoryless scalar quantizer), we emphasize that all our presented approaches are MSQ schemes that can be realized in hardware via switches, one-bit ADCs, and thermal noise diodes, e.g., see \cite[Figure 2]{yang2023plug} and \cite{robinson2019analog}.}

The software that implements ESPRIT, computes the quantizers in this paper, and reproduces all the numerical experiments located in this section are available on WL's Github repository.\footnote{https://github.com/weilinlimath/OneBitSubspaceDOA}

\subsection{An adversarial example}

In this example, we will consider an adversarial situation that illustrates the inconsistency of direct rounding for deterministic signals. We first select an arbitrary real $\bU\in\O^{p\times s}$ that will be fixed throughout the experiment. In our simulations, $\bU$ is a single random draw according to the Haar measure, and from performing simulations numerous times, we observe that its choice does not affect the following reported rates of approximation. 

In this experiment, we set $p=32$, $s=8$, and will vary $n$. For each $k\in [n]$, let $k'=k\mod s$, and we define the clean samples $\bx_1,\dots,\bx_n$  such that $\bx_k$ is the $k'$-th column of $\bU$. Although $\bx_1,\dots,\bx_n$ cycle through the same $s$ samples in a periodic manner, they obviously provide enough information to fully recover the subspace $\bU$. 

We draw real i.i.d.\ samples $\be_1,\dots,\be_n\in \calU(-\nu,\nu)^p$ with $\nu=0.01$, and let $\by_k=\bx_k+\be_k$ be the noisy unquantized samples. We consider five different quantizers: rectangular dithering, triangular dither with $b=2$ and $b=4$ bits, and direct rounding with $b=2$ and $b=4$ bits. Each quantizer uses the same $\lambda$ parameter that controls the range of inputs, which in this experiment, is chosen as $\lambda=2$. 

For quantized samples generated by either rectangular or triangular dithering, we have already discussed how to compute $\bU_n^\square$ and $\bU_n^\triangle$, as summarized in \cref{alg:subspace}. If $\bq_1,\dots,\bq_n$ are quantized measurements computed via direct rounding, we let $\bU_n^{R}$ be the leading left singular space of the matrix $[\bq_1,\dots,\bq_n]$. 

To make this comparison among these five quantizers fair, we normalize by the total bits used. A quantizer that uses $b$ bits per real scalar will use $bpn$ bits for $n$ unquantized measurements in $\R^p$. This means that rectangular, 2-bit triangular, and 2-bit direct rounding can utilize twice as many measurements compared to 4-bit triangular and 4-bit direct rounding. 

\begin{figure}[h]
	\centering
	\includegraphics[width=.5\textwidth]{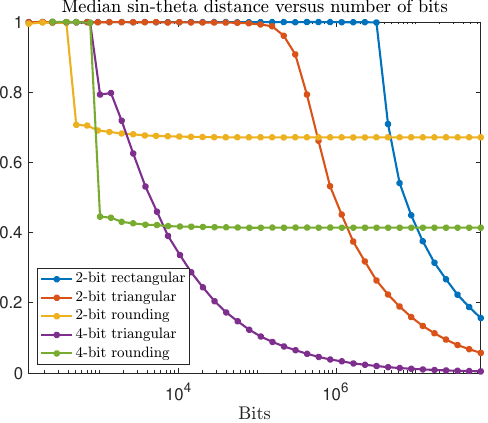}
	\caption{Results of experiment for adversarial example. }
	\label{fig:adversarial}
\end{figure}

We perform this routine for 100 trials. Note that the only stochastic component in this experiment is the noise since $\bx_1,\dots,\bx_n$ are fixed. We compute the median sin-theta distance versus the number of bits. The results are shown in Figure \ref{fig:adversarial}. We see that 2-bit and 4-bit direct rounding quickly saturate after enough measurements are taken. This is expected since $\bx_1,\dots,\bx_n$ periodically rotate through the same points, so direct rounding tends to map $\by_1,\dots,\by_n$ to the same collection of values, and consequently, does not take advantage of additional measurements. On the other hand, the three dithered quantizers can take advantage of repetitive measurements precisely due to their stochastic nature. For instance, even though $\by_k=\bx_k+\be_k$ and $\by_{k+s}=\bx_{k+s}+\be_{k+s}$ are close together, their quantized samples generated by dithered quantizers may be entirely different. 

Finally, one may consider other noise such as $\calN(\0,\nu^2\id_p)$. 
The rates in Figure \ref{fig:adversarial} do not change when $\calN(\0,\nu^2\id_p)$ noise is used instead, but we do not include this figure in the paper.

\subsection{Dependence on smallest singular value for rectangular dithering}

As discussed in \cref{sec:comparison}, the main appeal of quantization via triangular dithering over rectangular is an improved dependence on {$\kappa_s(\SigmaX)$}. It is natural to wonder if it is an artifact of our proofs for rectangular dithering, or whether it actually appears in examples.

Let us first explain the set-up of the following example. We will consider adjustable $\zeta\in(0,1)$ and a covariance matrix $\SigmaX\in \R^{p\times p}$ with rank exactly $r<p$ such that 
$$
\SigmaX=\diag(1,\zeta,\dots,\zeta,0,\dots,0), 
$$
i.e., $\zeta$ appears $r-1$ times, and 0 is repeated $p-r$ many times. Note that $\lambda_1(\SigmaX)=1$ and $\lambda_r(\SigmaX)=\zeta$. We consider the noiseless setting $\nu=0$, set $r=15$ and $p=20$, fix a rotation matrix $\bR$, e.g., drawn uniformly at random according to the Haar measure, and draw $\bx=\bR \bg$ where $\bg\sim N(\0,\SigmaX)$. $\bR$ is drawn once and fixed in this experiment.

Under the same simplifications as in \cref{sec:comparison}, ignoring all terms except those that depend on $n$ and $\lambda_r(\SigmaX)$, we obtain for $n$ sufficiently large
$$
\dist(\bU^\square_n,\bU)
\lesssim_{p,\lambda,K} \frac{1}{\lambda_r(\SigmaX) \sqrt n}.
$$
If we increase $n$ and enforce the scaling $\lambda_r(\SigmaX)=\zeta\simeq n^{-\beta}$ for $\beta>0$, then the upper bound for the sin-theta distance will:
\begin{enumerate}[(i)] \itemsep-2pt
	\item
	remain constant at the critical exponent $\beta = 1/2$,
	\item 
	decrease in $n$ when $\beta<1/2$,
	\item 
	eventually achieve the trivial upper bound of 1  when $\beta>1/2$.
\end{enumerate}  

\begin{figure}[h]
	\centering
	\includegraphics[width=0.5\textwidth]{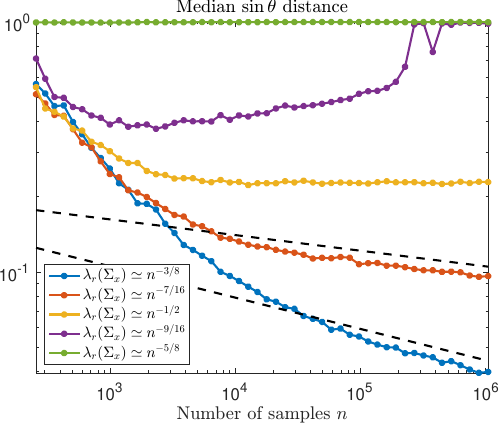}
	\caption{Dependence on smallest singular value for rectangular dithering. The black dashed lines are $1/(4n^{1/16})$ and $1/(4n^{1/8})$.}
	\label{fig:eigendependencerect}
\end{figure}

To numerically test whether the upper bound is sharp, we draw 100 samples from $\bx$ under the scaling $\lambda_r(\SigmaX)\simeq n^{-\beta}$ for $\beta \in \{3/8, \, 7/16, \, 1/2, \, 9/16, \, 5/8\}$. We display the median sin-theta distance in Figure \ref{fig:eigendependencerect}. The results of this experiment emphatically show that we cannot expect a better dependence on $\lambda_r(\SigmaX)$, such as $1/\sqrt{n \lambda_r(\SigmaX)}$, otherwise the sin-theta distance would decrease in $n$ under the scaling $\lambda_r(\SigmaX)\simeq n^{-\beta}$ for $\beta<1$. The numerically observed decay rates for $\beta=3/8$ and $\beta=7/16$ are consistent with the rate predicted by \cref{thm:subspacerectangle}.

\subsection{Dependence on smallest singular value for triangular dithering}

Here we verify an important claim from \cref{thm:subspacerectangletriangle} for the triangular dither. In high signal-to-noise situations where $\mu^2+\nu^2\leq \lambda_s(\SigmaX)$, we have for probability $1-e^{-c\alpha p}$,
$$
\dist(\bU_n^\triangle,\bU)
\leq C \sqrt{\frac {\nu^2+\mu^2}{\lambda_s(\SigmaX)}} \sqrt{\frac {\alpha p} n}
= \frac {C\sqrt{(\nu^2+\mu^2)\alpha p}}{\sigma_s(\bX_n)}. 
$$

To test the dependence on $\sigma_s(\bX_n)$, we consider the following sequence of deterministic signals. We pick 32 logarithmically spaced integers between $10^3$ and $10^7$. Letting $N$ denote any such number, define the vectors $\bx_1,\dots\bx_N\in \R^p$ for $p=8$, where
$$
\bx_k =  \begin{bmatrix}
	\, \cos\left( \frac{16 \pi (k-1)}N\right) &\sin\left( \frac{16 \pi (k-1)}N\right) &0 &\cdots &0
\end{bmatrix}^T. 
$$
That is, the first two components of $\bx_1,\dots\bx_N$ are equally spaced samples from eight rotations of the unit circle in $\R^2$, while the remaining entries are zero. Clearly, $\bx_1,\dots,\bx_N$ are samples from a subspace $\bU$ of dimension $r=2$. 

We draw i.i.d.\ samples $\be_1,\dots,\be_N\sim \calU(-\nu,\nu)^p$ for $\nu=0.01$ and set $\by_k=\bx_k+\be_k$ for each $k\in [N]$. For varying $n\in [N]$, we consider the data $\bX_n\in \R^{p\times n}$ whose columns consist of $\bx_1,\dots,\bx_n$. As $n$ increases, so does $\sigma_r(\bX_n)$. The samples $\by_1,\dots,\by_n$ are quantized using $b$-bit triangular dithers with $\lambda=2$ and $b\in \{2,4,6,8\}$. The approximation $\bU_n^\triangle$ to $\bU$ is found through \cref{alg:subspace}.

\begin{figure}[h]
	\centering
	\includegraphics[width=.5\textwidth]{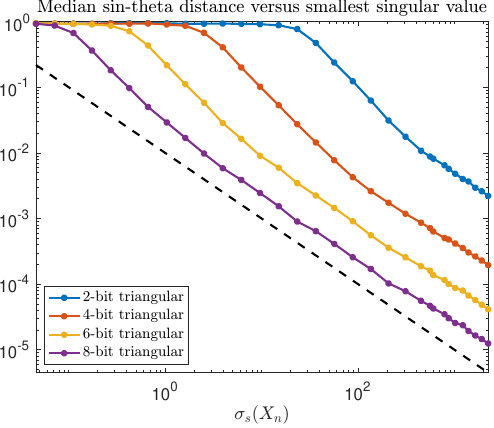}
	\caption{Dependence on smallest singular value for triangular dither, and the black dashed line is $C/\sigma_r(\bX_n)$.}
	\label{fig:eigendependence}
\end{figure}

We perform this routine 100 times and compute the median $\dist(\bU_n^\triangle,\bU)$. Note that the only stochastic component of this experiment is the noise, so we draw 100 independent realizations of $\be_1,\dots,\be_N$. The results shown in Figure \ref{fig:eigendependence} numerically verify our claim in \cref{thm:subspacerectangletriangle}. 

\subsection{Spectral estimation example: well-separated case}

Here, we consider the quantized spectral estimation problem in the well-separated case where $\Delta(\btheta)\geq \frac 2 p$ and for deterministic amplitudes. As discussed earlier, this is a highly favorable situation since \eqref{eq:wellsep} shows that $\sigma_s(\bPhi)\simeq\sigma_1(\bPhi) \simeq \sqrt p$. Provided that the assumptions in \cref{cor:espritrectangular,cor:esprittriangular} hold, the estimators $\btheta_n^\square$ and $\btheta_n^\triangle$ both decay at rates of $O(1/\sqrt n)$. 

To numerically verify these results, we set $p=32$ and $s=4$. We create an arbitrary $\btheta$ that satisfies $\Delta(\btheta)\geq \frac 2 p$ by setting $\theta_k= \frac{4k}p +  \tau_k$ where $\tau_1,\dots,\tau_s\in \calU(-\frac 1p,\frac 1p)$ are independent. Here, $\tau$ is only used as a random perturbation in order to avoid any number theoretic properties that may have some special interactions with the Fourier transform. Even though $\btheta$ is partly chosen in a stochastic manner, it is fixed throughout the subsequent experiment. We let the amplitudes $\ba_1,\dots,\ba_n$ cycle through the $s$ canonical basis vectors for $\R^s$,  where is $n$ is always chosen to be a multiple of $s$. In this case, $\SigmaA=\id$ and $\lambda_s(\SigmaA)=1$.

For each trial, we draw independent $\be_1,\dots,\be_N\sim \calU_\C(-\nu,\nu)^p$ where $\nu=0.01$, and form the noisy unquantized measurement $\by_k = \bPhi(\btheta)\ba_k + \be_k$ for each $k\in [N]$. These samples are quantized using five different quantizers: rectangular dithering with $4$ bits, triangular with $b=4$ and $b=8$ bits, and direct rounding with $b=4$ and $b=8$ bits. Each quantizer uses $\lambda=6$. 

To set up a fair comparison, we normalize by the total bits used. Recall that $\by_k$ is a complex vector, so the number of bits used is multiplied by a factor of 2; for example, $4$-bit rectangular dithering uses 2 bits per real scalar and 2 bits per imaginary scalar. A quantizer that uses $b$-bits per complex scalar entry uses $B=bpn$ bits for $n$ measurements in $\C^p$. Since $p$ is fixed throughout this experiment, the total number of bits $B$ is proportional to $n$. 

For each of these five quantizers, we generate quantized samples $\bq_1,\dots,\bq_n$ and then proceed to estimate $\btheta$ using their associated covariance estimators in ESPRIT (see \cref{alg:espritquan}). We repeat the above procedure for 100 trials and report the median matching distance error. Again, this is averaged over 100 draws of the noise, while $\btheta$ and the amplitudes are fixed. The results are shown in Figure \ref{fig:espritbits}. 

\begin{figure}[h]
	\centering
	\includegraphics[width=.5\textwidth]{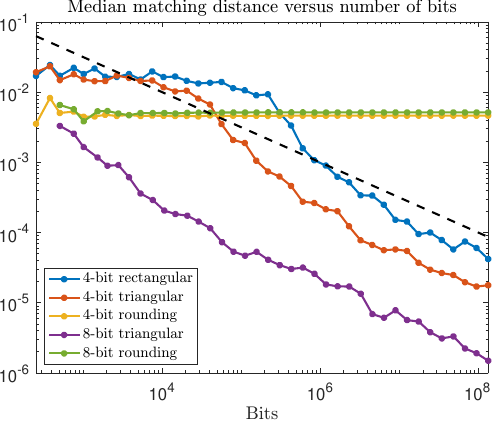}
	\caption{ESPRIT using five different quantizers. The black dashed line is $ 1/{\sqrt {\rm Bits}}$.}
	\label{fig:espritbits}
\end{figure}

As expected, ESPRIT is not a consistent estimator of $\btheta$ when given directly rounded measurements because the analog measurements $\{\bPhi(\btheta)\ba_k\}_{k\in [n]}$ cycle through the same $s$ elements in $\C^p$ and direct rounding is unable to take advantage of additional measurements. On the other hand, the simulations verify our claim that the dithered estimators are consistent and the matching distance error decays at a rate of $O(\frac 1 {\sqrt B})$ as $B\to\infty$, or equivalently, $O(\frac 1 {\sqrt n})$ as $n\to\infty$. Triangular dithering improves as we increase the number of bits since this effectively reduces the quantizer resolution $\mu$ in the term $\sigma^2=\mu^2+\nu^2$. 

This example may seem artificial since the amplitudes cycle through the canonical basis vectors. On the other hand, this is the ideal scenario for spectral estimation. When $\ba_k$ is a canonical basis vector, then we collect information from just a single frequency per time snapshot. This is in contrast to a more complicated and realistic scenario where all entries of the amplitude vector are nonzero, and in which case, we would collect a super-position of waves that would need be disentangled. Yet, even for this ideal situation, the performance of direct rounding quickly saturates, which illustrates a significant advantage for dithered quantizers when many samples can be collected. 

\subsection{Pushing the limits of quantization}

Here, we consider the most extreme scenarios that our theory offers. We again look at 4-bit rectangular and triangular dithering for spectral estimation. This is the fewest bits that our methods can use, since we must work with complex measurements, {see also Remark \ref{rem:ADC}}. We examine the super-resolution scenario whereby the minimum separation is allowed to be arbitrarily small and much smaller than $\frac 1 p$. In this setting, spectral estimation is highly sensitive to noise and errors from quantization. We consider the case of deterministic amplitudes for simplicity.

More concretely, we set $p=32$ and parameterize the minimum separation by varying $\epsilon\in [\frac 1 {32p}, \frac 1 p]$. Consider sources located at $\btheta = \{0,\epsilon, \frac 12\}$ so that $\Delta(\btheta)=\epsilon$ and $s=3$. This is an important example, which corresponds to the scenario whereby two frequencies/objects/sources are close together {and difficult} to distinguish, while a third frequency/object/source is far away from the others and serves as a decoy. In this scenario, we have $\sigma_s(\bPhi)\geq C \sqrt p \, (p \epsilon)$ for some numerical $C>0$, and this estimate is sharp in $\epsilon$ and $p$, see \cite{li2024multiscale}. {If one is interested in replicating the experiment with more complicated choices of} $\btheta$, {one can find a corresponding} expression for $\sigma_s(\bPhi)$ in \cite{li2024multiscale}. 

For large enough $N$, we select vectors $\ba_1,\dots,\ba_N\sim \calU_\C(-1,1)^s$ independently {and keep them fixed} through the rest of this experiment. We also vary the number of measurements $n$ while keeping $p$ fixed, so the total number of bits $B$ used for either 4-bit rectangular or triangular dithering is $4pn$. Since we treat $p$ as a constant, $B\simeq n$. For each trial, we generate $\be_1,\dots,\be_N\sim \calU_\C(-\nu,\nu)^p$ independently with $\nu=0.01$. For $n\leq N$, the unquantized complex samples are $\by_1,\dots,\by_n$, where $\by_k=\bPhi(\btheta)\ba_k + \be_k$ for each $k\in [n]$. These are quantized using either 4-bit rectangular or 4-bit triangular dithering with parameter $\lambda = 5$.

\begin{figure}[h]
	\centering
	\includegraphics[width=.45\textwidth]{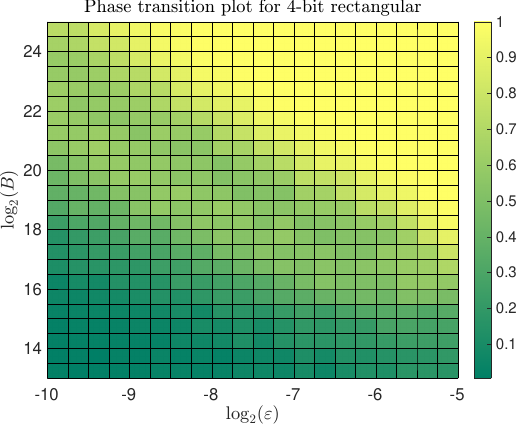} \ \ 
	\includegraphics[width=.45\textwidth]{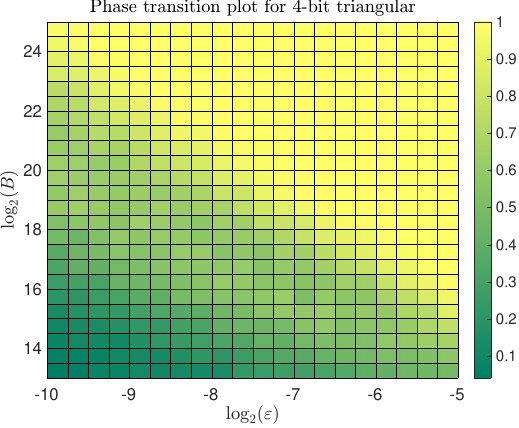}
	\caption{Phase transition plots.}
	\label{fig:phase}
\end{figure}

For a single trial, we say ESPRIT is successful using rectangular (resp. triangular) dithering if $\md(\btheta_n^\square,\btheta)\leq \frac 14 \Delta(\btheta) = \frac 14 \epsilon$ (resp. $\md(\btheta_n^\triangle,\btheta)\leq \frac 14 \epsilon$). Here, the $\frac 14$ factor was chosen arbitrarily, but since $\epsilon$ is the minimum separation of $\btheta$, success should be quantified relative to $\epsilon$. For each pair $(\epsilon,B)$ we repeat the above procedure 500 times for each quantizer and report the fraction of trials for which ESPRIT is successful. The noise is the only stochastic component of this experiment. 

Our results are shown in Figure \ref{fig:phase}. For both quantizers, we see that there is a phase transition region {whereby the} method appears to go from failure to success. This region is a strip under logarithmic scale. We say $(\epsilon,B)$ is usually successful if at least 95\% of 500 trials for that pair of parameters are successful. For each $\epsilon$, we find the smallest $B$ for which $(\epsilon,B)$ is usually successful, and fit these data points with the best possible line in the least squares sense. The slopes of these lines are $-1.6814$ and $-1.7193$ for 4-bit rectangular and 4-bit triangular, respectively. Hence, the empirical results suggest that $B \geq C (\frac{1}{\epsilon})^{1.6814}$ and $B \geq C (\frac{1}{\epsilon})^{1.7193}$ are enough for rectangular and triangular dithering to succeed at least 95\% over these trials.

Our theorems cannot fully explain these results. For triangular dithering, using \cref{thm:esprittriangular} together with $\sigma_s(\bPhi) \geq C \sqrt p \, (p \epsilon)$ and $B=4np$, with fixed probability,
$$
\md(\btheta_n^\triangle,\btheta)
\lesssim_{s,p,\nu,\lambda_s(\SigmaA)} \frac{1}{\epsilon \sqrt B},
$$
{which would require $B \ge C(\frac{1}{\varepsilon})^4$ for usually successful recovery}.
{For rectangular dithering the gap is even worse.} Using \cref{thm:espritrectangular} and greatly simplifying the resulting expression, with fixed probability,
$$
\md(\btheta_n^\square,\btheta)
\lesssim_{s,p,\nu,\lambda_s(\SigmaA)} \frac{1}{\epsilon^2 \sqrt{B}},
$$
{which would require $B \ge C(\frac{1}{\varepsilon})^6$ for usually successful recovery}.

\section{Proofs}
\label{sec:Proofs}

In this section, we provide the proofs of our main results.

\subsection{Proof of Lemma \ref{lem:PSestviaCovEst}}
\label{sec:PSestviaCovEst}

Thanks to the trivial estimate $\dist(\hat{\bU},\bU)\leq 1$, it suffices to prove \eqref{eq:wedin2} in the case that
\begin{align}
\label{eq:ReducedCase}
    \|\hat{\SIGMA}-\SIGMA_{\by}\|
	\leq \frac 1 {1+\sqrt 2}  \, \big(\lambda_s(\SigmaX)-\lambda_{s+1}(\SigmaX)\big).
\end{align}
Since $\SigmaY=\SigmaX+\nu^2\id$, it follows that for every $k\in [p]$, 
\begin{equation}
\label{eqn:specShift}
\lambda_k(\SIGMA_{\by})=\lambda_k(\SigmaX)+\nu^2.
\end{equation} 
By assumption, $\lambda_s(\SigmaX)>\lambda_{s+1}(\SigmaX)$ and hence the leading $s$-dimensional eigenspace of $\SigmaY$ is unique. By Weyl's inequality for Hermitian matrices,  
	\begin{align*}
		\lambda_k(\hat{\SIGMA})
		&\geq \lambda_k(\SIGMA_{\by})-\|\hat{\SIGMA}-\SIGMA_{\by}\| \\
		\lambda_k(\hat{\SIGMA}) 
		&\leq \lambda_k(\SIGMA_{\by})+\|\hat{\SIGMA}-\SIGMA_{\by}\|.
	\end{align*}
	Using these inequalities together with \eqref{eq:ReducedCase} and \eqref{eqn:specShift} we find
	\begin{align*}
		\lambda_s(\hat{\SIGMA})-\lambda_{s+1}(\hat{\SIGMA})
		&\geq 
		\lambda_s(\SigmaX)-\lambda_{s+1}(\SigmaX)- 2\|\hat{\SIGMA}-\SIGMA_{\by}\|>0,
	\end{align*}
	so that in particular the leading $s$-dimensional eigenspace of $\hat{\SIGMA}$ is unique.
	The Davis-Kahan theorem \cite{davis1970rotation} and \eqref{eq:ReducedCase} now yield
	\begin{align*}
	\dist({\hat \bU},\bU)
	&\leq \frac{\sqrt 2 \, \|\hat{\SIGMA}-\SIGMA_{\by}\|}{\lambda_s(\SigmaX)-\lambda_{s+1}(\SigmaX)-\|\hat{\SIGMA}-\SIGMA_{\by}\|} \\
	&\leq (1+\sqrt 2) \, \frac{\|\hat{\SIGMA}-\SIGMA_{\by}\|}{\lambda_s(\SigmaX)-\lambda_{s+1}(\SigmaX)}.     
	\end{align*}

\subsection{Proof of Theorem \ref{thm:subspacerectangle}}
\label{sec:subspace}

In this section, we extend the analysis in \cite{dirksen2022covariance,yang2023plug} to bound $\| \SIGMADITHuniform_n - \SigmaY \|$ under Assumption \ref{assump:main}, see Theorem~\ref{thm:CE_ErrorBound} below. Theorem \ref{thm:subspacerectangle} then follows from combining Lemma~\ref{lem:PSestviaCovEst} with Theorem~\ref{thm:CE_ErrorBound}. 

Let us begin with some technical observations.
The following result shows that $\tilde\SIGMA_n'$ and $\SIGMADITHuniform_n$ are unbiased estimators of $\SIGMA_{\by}$ if the dithering range $\lambda$ dominates the range of the underlying distributions.

\begin{lemma} \label{lem:linftyBiasEstBounded}
    Let $\lambda > 0$ and let $\bu \in \mathbb{C}^p$ be a random vector with $\| \bu \|_\infty \le \lambda$ a.s. and covariance matrix $\SIGMA_{\bu} = \E(\bu\bu^*)$. Let $\bq^\square,\dot \bq^\square \in \F^p$ be defined via $\bq^\square = \sign_\F \big( \bu + \tauv^\square \big)$ and $\dot \bq^\square = \sign_\F \big( \bu + \dot \tauv^\square \big)$, where $\tauv^\square,\dot \tauv^\square\sim \calU_\F (-\lambda,\lambda)^p$ and assume that $\bu,\tauv^\square,\dot \tauv^\square$ are independent.
    Then,
    \begin{align*}
        \E \big( \lambda^2 \bq^\square(\dot \bq^\square)^* \big) = \SIGMA_{\bu}
    \end{align*}
    and
    \begin{align*}
       \E \big( \lambda^2 \bq^\square(\bq^\square)^* \big) = \SIGMA_{\bu} - \diag (\SIGMA_{\bu}) + c_\F \lambda^2 \id_M,
    \end{align*}
    where $c_\R = 1$ and $c_\C = 2$.
\end{lemma}

The proof of this lemma is a straightforward modification of \cite[Lemma 2]{yang2023plug}.
The following result slightly extends \cite[Theorem 4]{dirksen2022covariance}. We include a proof for the reader's convenience. 
\begin{theorem}
\label{thm:CE_ErrorBound}
    Suppose \cref{assump:main} holds.
    Letting $C_\infty$ and $\Cop$ be the quantities defined in \eqref{eq:Cy_defs}, the following statements hold:
    \begin{enumerate}
        \item[(i)] 
        If $\lambda^2 \gtrsim_K \log(n) C_{\infty}$, then with probability at least $1-e^{-t}$
        \begin{align}
        \label{eq:CE_ErrorBound_Subgaussian}
            \| \SIGMADITHuniform_n - \SIGMA_{\by} \|
            \lesssim_K ( \lambda \Cop^{1/2}  + \lambda^2 ) \sqrt{\frac{p (\log(p) + t)}{n}} + \lambda^2 \frac{p (\log(p) + t)}{n} .
        \end{align}
        In particular, if $\lambda^2 \simeq_K \log(n) C_{\infty}$, then
        \begin{align*}
            \| \SIGMADITHuniform_n - \SIGMA_{\by} \|
            \lesssim_K \log(n) \sqrt{ \Cop C_{\infty} \frac{p (\log(p) + t)}{n}} + \log(n) C_{\infty} \frac{p (\log(p) + t)}{n}.
        \end{align*}
        \item[(ii)] If the distributions of the samples $\by_k$ are bounded, then $\lambda$ can be chosen independent of $n$, i.e., if $\| \by_k \|_\infty \le M$ a.s. for each $k \in [n]$ and $\lambda \ge M$, we have with probability $1-e^{-t}$ that \eqref{eq:CE_ErrorBound_Subgaussian} holds. 
    \end{enumerate}
\end{theorem}
\begin{proof}[Proof of Theorem \ref{thm:CE_ErrorBound}] 
Under Assumption \ref{assump:main}, the samples $\by_k$ are independent and $K$-subgaussian.
We only prove (ii). For random samples $\bx_1,\dots,\bx_n \sim \bx$ and $\F = \R$, the statement in (i) is equivalent to \cite[Theorem 4]{dirksen2022covariance}. The extension to $\F = \C$ and deterministic samples $\bx_1,\dots,\bx_n$ works analogously to (ii) but using a complex analog of \cite[Lemma 17]{dirksen2022covariance} instead of Lemma~\ref{lem:linftyBiasEstBounded}. 

    \textit{Proof of (ii):} First note that
    \begin{align*}
        \| \SIGMADITHuniform_n - \SIGMA_{\by} \|
        \le \| \hat{\SIGMA}'_n - \SigmaY \|
        = \left\| \sum_{k=1}^n \THETA_k \right\|,
    \end{align*}
    where we define
    \begin{align*}
        \THETA_k 
        &:= \frac{\lambda^2}{n} \left( \bq_k^\square (\dot \bq_k^\square)^* - \E \big( \bq_k^\square (\dot \bq_k^\square)^* \big) \right) \\
        &= \frac{1}{n} \left( \lambda^2 \bq_k^\square (\dot \bq_k^\square)^* - \E \big( \by_k\by_k^* \big) \right).
    \end{align*}
    and use Lemma \ref{lem:linftyBiasEstBounded} in the second equality. By our assumption on $\by_k$, the $\THETA_k$ are independent matrices with mean zero. To conclude we will apply \cite[Theorem 6.2]{Dir14}, cf.\ \cite[Theorem 10]{dirksen2022covariance}.\\
    First, we observe that
    \begin{align*}
        \| \THETA_k \| 
        = \frac{\lambda^2}{n} \left\| \bq_k^\square (\dot \bq_k^\square)^* - \E \big( \bq_k^\square (\dot \bq_k^\square)^* \big) \right\|
        \lesssim \frac{\lambda^2 p}{n},
    \end{align*}
    where we used that $\max_{k} \{ \| \bq_k^\square \|_\infty, \| \dot \bq_k^\square \|_\infty \} \le \sqrt{2}$.
    Second,
    \begin{align*}
        &\Big\|\Big(\sum_{k=1}^n \mathbb{E}(\THETA_k^*\THETA_k)\Big)^{1/2}\Big\| \\
        &\le \frac{\lambda^2}{n} \left(
        \sum_{k=1}^n \Big\| \E \big( \dot \bq_k^\square (\bq_k^\square)^* \bq_k^\square (\dot \bq_k^\square)^* \big) - \E \big( \dot \bq_k^\square (\bq_k^\square)^* \big) \E \big( \bq_k^\square (\dot \bq_k^\square)^* \big) \Big\| \right)^{1/2} \\
        &\le \frac{2\lambda^2}{n} \left(
        \sum_{k=1}^n \Big\| \E \big( \dot \bq_k^\square (\bq_k^\square)^* \bq_k^\square (\dot \bq_k^\square)^* \big) \Big\| \right)^{1/2} \\
        &\le 4\lambda^2 \frac{\sqrt{p}}{n} \left(
        \sum_{k=1}^n \Big\| \E \big( \dot \bq_k^\square (\dot \bq_k^\square)^* \big) \Big\| \right)^{1/2} \\
        &\le 4\lambda \frac{\sqrt{p}}{n} \left(
        \sum_{k=1}^n \Big\| \E \big( \by_k\by_k^* \big) - \diag \big( \E \big( \by_k\by_k^* \big) \big) + c_\F \lambda^2 \id_p \Big\| \right)^{1/2} \\
        &\le 4\lambda \sqrt{\frac{p}{n}} \max_{k \in [n]} \Big\| \E \big( \by_k\by_k^* \big) - \diag \big( \E \big( \by_k\by_k^* \big) \big) + c_\F \lambda^2 \id_p \Big\|^{1/2} \\
        &\lesssim \sqrt{\frac{p}{n}} \left( \lambda \max_{k \in [n]} \| \E \big( \by_k\by_k^* \big) \|^{1/2} + \lambda^2 \right),
    \end{align*}
where we used \eqref{eq:Kadison} in the second line, the observation that $\| \bq_k^\square \|_2^2 \le 2p$ in the third line, and Lemma \ref{lem:linftyBiasEstBounded} in the fourth line. Analogously,
    \begin{align*}
        \Big\|\Big(\sum_{k=1}^n \mathbb{E}(\THETA_k \THETA_k^*)\Big)^{1/2}\Big\|
        \lesssim \sqrt{\frac{p}{n}} \left( \lambda \max_{k \in [n]} \| \E \big( \by_k\by_k^* \big) \|^{1/2} + \lambda^2 \right),
    \end{align*}

    Combining the previous three estimates and applying \cite[Theorem 6.2]{Dir14} thus leads to
    \begin{align*}
        \Big(\E \Big\|\sum_{k=1}^n {\THETA}_k \Big\|^q \Big)^{1/q}
        \lesssim \sqrt{q} \sqrt{\frac{p}{n}} \left( \lambda \max_{k \in [n]} \| \E \big( \by_k\by_k^* \big) \|^{1/2} + \lambda^2 \right)
        + q \lambda^2 \frac{p}{n},
    \end{align*}
    for any $q \ge \log(p)$. By \cite[Lemma 13]{dirksen2022covariance}, this implies that \eqref{eq:CE_ErrorBound_Subgaussian} holds with probability at least $1-e^{-t}$.
\end{proof}

\subsection{Proof of Lemma \ref{lem:triquannoise}}
\label{sec:triquannoise}

{
For $\by_k = \bx_k + \be_k$, let us first define $\bgamma_k = Q_\mu(\by_k+\tauv_k^\triangle) - \by_k$. Since $\tauv_k^\triangle$ has i.i.d.\ entries that follow the triangular distribution, we derive from \cite[Theorem 2]{gray1993dithered} that, for $k \in [n]$ and $\ell \in [p]$, there exist $w_{k,\ell}\sim \calU_\R (-\mu,\mu)$ that are i.i.d.\ and independent of the other random variables, such that the entries $\gamma_{k,\ell}$ of $\bgamma_k$ satisfy
\begin{enumerate}[{\rm (a)}]
	\item 
	$
	\E(\gamma_{k,\ell}\mid \by_k)
	=\E(\gamma_{k,\ell})
	=\E \left(\tau^\triangle_{k,\ell}+w_{k,\ell} \right), 
	$ 
    \item 
	$
	\E(\gamma_{k,\ell}^2\mid \by_k)
	=\E(\gamma_{k,\ell}^2)
	=\E \left( \left(\tau^\triangle_{k,\ell}+w_{k,\ell} \right)^2 \right), 
	$ 
	\item $\gamma_{k,\ell}$ and $\gamma_{k,\ell'}$ are uncorrelated for $\ell\neq \ell'$.
\end{enumerate}	
Since $\bxi_k = \bgamma_k + \be_k$, we get from (a) that
\begin{align*}
    \E \xi_{k,\ell} 
    = \E\gamma_{k,\ell} + \E e_{k,\ell} 
    = (\E \tau^\triangle_{k,\ell} + \E w_{k,\ell} ) + 0
    = 0
\end{align*}
and that
\begin{align*}
    \E \xi_{k,\ell}^2 
    &= \E \gamma_{k,\ell}^2 + 2 \E \gamma_{k,\ell} e_{k,\ell} + \E e_{k,\ell}^2
    = (\E (\tau^\triangle_{k,\ell})^2 + \E w_{k,\ell}^2 ) + 2 \E_{\by_k} \E (\gamma_{k,\ell} e_{k,\ell} \mid \by_k) + \nu^2 \\
    &= \mu^2 + \nu^2,
\end{align*}
where we used that by $\E\gamma_{k,\ell} = 0${, it holds that}
\begin{align*}
    \E_{\by_k}\E(\gamma_{k,\ell} e_{k,\ell} \mid \by_k)
    = \E_{\by_k}(\E(\gamma_{k,\ell} \mid \by_k) \cdot \E(e_{k,\ell} \mid \by_k)) 
    = \E\gamma_{k,\ell} \cdot \E_{\be_k}\E(e_{k,\ell} \mid \by_k)
    = 0.
\end{align*}
Furthermore, together with our assumption that $\be_k$ has uncorrelated entries, (b) implies that $\bxi_k$ has uncorrelated entries.}

{
It remains to show that the $\bxi_k$ are subgaussian. Since $|Q_\mu(t)-t|\leq \mu$ for any $t\in\R$, we know from \cite[Theorem 1]{gray1993dithered} that $\bgamma_k-\tauv_k^\triangle$ has independent entries bounded by $\mu$. Hence, for any $\bu\in \R^p$, we can use $\bxi_k = \bgamma_k + \be_k$ and Hoeffding's inequality to see that
\begin{align*}
	&\|\langle \bxi_k,\bu\rangle\|_{\psi_2} \\
    &\leq \|\langle \bgamma_k-\tauv_k^\triangle,\bu\rangle\|_{\psi_2} + \|\langle \tauv_k^\triangle,\bu\rangle\|_{\psi_2} + \|\langle \be_k,\bu\rangle\|_{\psi_2} \\
	& \lesssim K(\nu+\mu) \|\bu\|_2 \lesssim_{K} \|\langle \bxi_k,\bu\rangle\|_{L^2}.
\end{align*}
}

\subsection{Proof of \cref{thm:subspacerectangletriangle}(i)}
\label{sec:CZvariation}

The proof of Theorem \ref{thm:subspacerectangletriangle}(i) has a similar structure as the proof of Cai-Zhang \cite[Theorem 3]{cai2018rate}. In fact, the only stochastic component in the proof of \cite[Theorem 3]{cai2018rate} is \cite[Lemma 4]{cai2018supplement}, i.e., Lemma 4 in the supplement to \cite{cai2018rate}. This lemma assumes that the noise has i.i.d.\ subgaussian entries, so in order to derive Theorem \ref{thm:subspacerectangletriangle}(i) it suffices to find a suitable substitute to \cite[Lemma 4]{cai2018supplement} under our weaker hypothesis on the noise. 

To avoid repetitiveness, we only provide rigorous proofs for the parts that are significantly different. For the reader's convenience in comparing the respective statements, in the remaining section we will switch to the notation and setup as in \cite{cai2018rate,cai2018supplement}. To do this, we let $\sigma^2=c_{\F}^2(\nu^2+\mu^2)$ and consider the model
$$
\frac 1 \sigma (\bQ_n^\triangle)^* = \frac 1 \sigma \bX_n^*+ \frac 1 \sigma \bXi_n^*, 
$$
which is a transformation of \eqref{eq:Qtri}. Taking the Hermitian {on both sides} swaps the right and left singular spaces. We have $\sigma_r(\sigma^{-1}\bX_n^*) = \sigma^{-1} \sigma_r(\bX_n)$ and that $\lambda_r(\SigmaX)= n^{-1} \sigma_r^2(\bX_n)$. Recalling the properties of $\bXi_n$ from \cref{lem:triquannoise}, we see that $\sigma^{-1} \bXi_n^*$ has independent subgaussian rows with $\psi_2$-norm bounded by a universal constant, and the entries of $\bXi$ are uncorrelated, and have mean zero and unit variance. 

From this rescaling, we see that \cref{thm:subspacerectangletriangle}(i) is a consequence of the following theorem which is a high probability version of \cite[Theorem 3]{cai2018rate} generalized to our noise setting. One just has to identify $p_1$ with $n$, $p_2$ with $p$, $r$ with $s$, $\tau$ with $K$, $\bX$ with $\frac{1}{\sigma} \bX_n^*$, $\bZ$ with $\frac{1}{\sigma} \bXi_n^*$, $\bY$ with $\frac{1}{\sigma} (\bQ_n^\triangle)^*$, $\bV$ with $\bU$, and $\widehat\bV$ with $\bU_n^\triangle$.

\begin{theorem}
	\label{thm:CZvarHP}
	Let $\tau >0$, and $p_1,p_2,r\in \N$ with $r\leq \min\{p_1,p_2\}$. Suppose $\bX\in \F^{p_1\times p_2}$ has rank $r$ and the rank $r$ truncated SVD of $\bX$ is $\bU \bS \bV^*$. Assume $\bZ\in \F^{p_1\times p_2}$ is a random matrix whose rows $\bz_1^*,\dots,\bz_{p_1}^*$ are independent, isotropic, $\tau$-subgaussian random vectors. Let $\bY:=\bX+\bZ$ and $\hat \bU \hat \bS \hat \bV^*$ be the rank $r$ truncated SVD of $\bY$. There is a an absolute constant $C>0$ and constants $\alpha_\tau\geq 0$ and $c_\tau \geq 0$ depending only on $\tau$ such that for any $\alpha\geq \alpha_\tau$, with probability at least $1-e^{-c_\tau \alpha p_2}$,
	$$\|\sin(\bV,\hat \bV)\|^2 \leq \min \left\{ 1,\frac{ C\alpha \, p_2 \, (\sigma_r^2(\bX)+p_1)}{\sigma_r^4(\bX)}\right\}.$$  
\end{theorem}

We follow the same steps as in the proof of \cite[Theorem 3]{cai2018rate}, which requires that  $\bZ$ additionally has i.i.d.\ entries. 
The main point of deviation is \cref{lem:CZvar} below, which serves as a suitable substitute to \cite[Lemma 4]{cai2018supplement} under our weaker assumptions on $\bZ$.
We first recall a well known $\epsilon$-net result. 

\begin{lemma}
	\label{lem:epnet}
	There are absolute constants $C,c>0$ such that the following holds. Suppose $\bA\in \F^{r_1\times r_2}$ is a random matrix. Then for all $t>0$,
	$$
	\P(\|\bA\|\geq ct)
	\leq C^{r_1+r_2} \max_{\substack{\bu\in \F^{r_1}, \, \|\bu\|_2=1 \\ \bv\in \F^{r_2}, \, \|\bv\|_2=1}} \P(|\bu^* \bA \bv|\geq t).
	$$
\end{lemma}

\begin{proof}
	The real case of this lemma is well known, see also \cite[Lemma 5]{cai2018supplement}. For the complex case, we split $\bA=\bA_1+\i \bA_2$ into its real and imaginary components. First note that $\|\bA\|\leq 2\|\bA_1\|+2\|\bA_2\| \leq 4\max\{\|\bA_1\|,\|\bA_2\|\}$, where $\|\bA_1\|$ and $\|\bA_2\|$ are the $\ell^2\to\ell^2$ operator norms acting on real vectors. 
    Thus, for an absolute constant $c>0$ and any $t>0$, we have
		$$
		\P(\|\bA\|\geq 4ct)
		\leq \P(\|\bA_1\|\geq ct)+\P(\|\bA_2\|\geq ct). 
		$$
		The final two quantities can be controlled using the $\epsilon$-net statement for real matrices applied to $\bA_1$ and $\bA_2$. Hence, 
		$$
		\P(\|\bA\|\geq 4ct)
		\leq C^{r_1+r_2} \sum_{k=1}^2 \, \max_{\substack{\bu\in \R^{r_1}, \,  \|\bu\|_2=1 \\ \bv\in \R^{r_2}, \, \|\bv\|_2=1}} \P(|\bu^* \bA_k \bv|\geq t). 
		$$
		Finally, if $\bu,\bv$ are real, a calculation shows that $|\bu^* \bA \bv|\geq \max \{ |\bu^* \bA_1 \bv|, |\bu^* \bA_2 \bv|\}$, which establishes that for $j \in \{1,2\}$,
		\begin{align*}
		&\max_{\substack{\bu\in \R^{r_1}, \, \|\bu\|_2=1 \\ \bv\in \R^{r_2}, \, \|\bv\|_2=1}} \P(|\bu^* \bA_j \bv|\geq t) \\
		&\leq \max_{\substack{\bu\in \R^{r_1}, \, \|\bu\|_2=1 \\ \bv\in \R^{r_2}, \, \|\bv\|_2=1}} \P(|\bu^* \bA \bv|\geq t) \\
		&\leq \max_{\substack{\bu\in \C^{r_1}, \, \|\bu\|_2=1 \\ \bv\in \C^{r_2}, \, \|\bv\|_2=1}} \P(|\bu^* \bA \bv|\geq t).   
		\end{align*}
		Combining the above completes the proof.
\end{proof}

Next, we establish some basic concentration bounds related to the matrix $\bZ$.

\begin{lemma}
	\label{lem:help1}
	There are absolute constants $C,c>0$ such that the following holds. Assume $\bZ\in \F^{p_1\times p_2}$ is a random matrix whose rows $\bz_1^*,\dots,\bz_{p_1}^*$ are independent, isotropic, $\tau$-subgaussian random vectors. For any fixed matrices $\bA\in \F^{r_1\times p_1},\bB\in \F^{p_2\times r_2}$ and $\bC\in \F^{p_2\times r_1}$, we have for all $t>0$, 
	\begin{align*}
	&\P(\|\bA \bZ \bB\|\geq t)
	\leq C^{r_1+r_2} \exp \left( - \frac{c t^2}{\tau^2 \|\bA\|^2 \|\bB\|^2 } \right), \\
	&\P(\|\bC^*(\bZ^*\bZ-p_1\id_{p_2}) \bB \|\geq t) \\
	&\leq C^{r_1+r_2} \exp\left( -c \min\left\{ \frac t{\tau^2 \|\bB\| \|\bC\|}, \, \frac {t^2}{p_1 \tau^4 \|\bB\|^2 \|\bC\|^2} \right\}\right). 
	\end{align*}
\end{lemma} 

\begin{proof}
	Fix unit vectors $\bu\in \F^{r_1},\bv\in \F^{r_2}$. Observe that 
	$
	\bu^* \bA \bZ \bB \bv
	= \tr(\bZ \bB \bv\bu^* \bA)
	= \sum_{k=1}^{p_1} \bz_k^* \bw_k, 
	$
	where we let $\bw_k$ denote the $k$-th column of $\bB \bv\bu^* \bA$. Since the rows of $\bZ$ are independent and the entries of $\bZ$ have mean zero, the previous expression is a sum of $p_1$ independent mean zero random variables. Since $\bz_k$ is $\tau$-subgaussian and $\bB \bv\bu^* \bA$ has rank one,   
	\begin{align*}
	\sum_{k=1}^{p_1} \|\bz_k^* \bw_k\|_{\psi_2}^2
	&\leq \tau^2 \sum_{k=1}^{p_1} \|\bw_k\|_2^2 \\
	&\leq \tau^2 \|\bB \bv\bu^* \bA\|_F^2
	\leq \tau^2 \|\bA\|^2 \|\bB\|^2. 
	\end{align*}
	Thus, by Hoeffding's inequality (see, e.g., \cite[Theorem 2.6.2]{vershynin2018high}), we get 
	$$
	\P(|\bu^* \bA \bZ \bB \bv|>t)\leq 2\exp \left(- \frac{c t^2}{\tau^2 \|\bA\|^2 \|\bB\|^2 } \right).
	$$
	Using the $\epsilon$-net result \cref{lem:epnet}, we obtain the first claimed inequality.  
	
	For the second part of this lemma, using $\E(\bz_k \bz_k^*)=\id_{p_2}$ yields
	\begin{align*}
		&\bu^*\bC^*(\bZ^*\bZ-p_1\id_{p_2}) \bB \bv \\
		&=\sum_{k=1}^{p_1} \bu^*\bC^*(\bz_k \bz_k^* - \E \bz_k \bz_k^*) \bB \bv\\
		&=\sum_{k=1}^{p_1} \Big( \bz_k^* \bB \bv \bu^*\bC^* \bz_k - \E(\bz_k^* \bB \bv \bu^*\bC^* \bz_k) \Big). 
	\end{align*}
	Due to the assumption that the rows of $\bZ$ are independent, the right hand side is a sum of independent, mean-zero random variables. Next, we see that 
	\begin{align*}
	&\|\bz_k^* \bB \bv \bu^*\bC^* \bz_k\|_{\psi_1} \\
	&\leq \|\bz_k^* \bB \bu\|_{\psi_2} \|\bu^* \bC^* \bz_k\|_{\psi_2} 
    \leq \tau^2 \|\bB\| \|\bC\|.    
	\end{align*}
	Thus, by Bernstein's inequality (see, e.g., \cite[Theorem 2.8.1]{vershynin2018high}), 
	\begin{align*}
	&\P(|\bu^*\bC^*(\bZ^*\bZ-p_1\id_{p_2}) \bB \bv|>t) \\
	&\leq 2\exp\left( -c \min\left\{ \frac t{\tau^2 \|\bB\| \|\bC\|}, \, \frac {t^2}{p_1 \tau^4 \|\bB\|^2 \|\bC\|^2} \right\}\right).     
	\end{align*}
	The $\epsilon$-net result in \cref{lem:epnet} completes the proof.
\end{proof}

From Lemma \ref{lem:help1}, we further derive the following concentration bounds. We use some additional notation to simplify the presentation. For a matrix $\bA$, let $P_{\bA}$ be the orthogonal projection onto the range of $\bA$. For $\bV \in \F^{p_2\times r}$ with orthonormal columns, let $\bV_\perp\in \F^{p_2\times (p_2-r)}$ be a matrix whose columns form an orthonormal basis for the orthogonal complement of the subspace spanned by $\bV$. 

\begin{lemma}
	\label{lem:CZvar}
	Under the same notation and assumptions of \cref{thm:CZvarHP}, there is an absolute $C>0$ and constants $C_\tau, c_\tau>0$ depending only on $\tau$ such that for any $t>0$,
    \begin{equation} \label{eq:CZvar1} 
        \begin{split}
        &\P \big(\sigma_r^2(\bY\bV)\leq (\sigma_r^2(\bX) + p_1)(1-t)\big) \\
		&\qquad \leq \exp\left(Cr - c_\tau (\sigma_r^2(\bX) + p_1) \min \{t,t^2\}\right),  
        \end{split}
    \end{equation}
    and 
    \begin{equation}\label{eq:CZvar2} 
        \begin{split}
        &\P\big(\sigma_{r+1}^2(\bY)\geq p_1(1+t)\big) \\
		&\qquad \leq \exp\left(C p_2-c_\tau p_1 \min\{t,t^2\}\right).     
        \end{split}
    \end{equation} 
	If additionally, $\sigma_r(\bX)\geq C_\tau p_2$, then for any $t>0$, 
	\begin{equation} \label{eq:CZvar3} 
    \begin{split}
    &\P \left(\|P_{\bY\bV} \bY\bV_\perp\| > t \right)  \\
	&\leq \exp\left(C p_2-c_\tau \min\left\{t^2, t\sqrt{\sigma_r^2(\bX)+p_1} \right\}\right) \\
    &\qquad +\exp\left(-c_\tau( \sigma_r^2(\bX)+p_1)\right). 
    \end{split}
	\end{equation}
\end{lemma}

\begin{proof}
	We use the shorthand notation $\sigma_k:=\sigma_k(\bX)$. Throughout, $C,c$ are absolute constants while $C_\tau,c_\tau$ are constants that only depend on $\tau$. Their values may change from one line to another. Define the $r\times r$ diagonal matrix, 
	$$
	\bM = \diag\left( (\sigma_1^2+p_1)^{-1/2}, \dots, (\sigma_r^2+p_1)^{-1/2}\right). 
	$$
	Setting $\bY_1:=\bY \bV$, we see that 
	$$
	\E (\bY^*_1 \bY_1) = \bS^2 + p_1 \id_r \andspace \E (\bM^* \bY^*_1 \bY_1 \bM) = \id_r.
	$$
	
	\noindent \underline{Proof of \eqref{eq:CZvar1}}. Observe that 
	\begin{align*}
	\label{eq:help1}
	\sigma_r^2(\bY_1)
	&\geq \sigma_r^2(\bM^{-1})\sigma_r^2(\bY_1 \bM) \\
	&\geq (\sigma_r^2 + p_1)\big( 1-\|\bM^* \bY_1^* \bY_1 \bM-\id_r\|\big). 
	\end{align*}
	Some algebraic manipulations show that
	\begin{align*}
	&\bM^* \bY_1^* \bY_1 \bM-\id_r \\
	&=\bM^* \bV^*\bX^* \bZ \bV \bM +\bM^* \bV^*\bZ^* \bX \bV \bM \\
    &\quad +\bM^* \bV^* (\bZ^*\bZ- p_1 \id_{p_2})\bV \bM.
	\end{align*}
    Note that $\|\bV \bM\| = \|\bM\| \leq (\sigma_r^2+p_1)^{-1/2}$ and $\|\bX \bV \bM\| =  \|\bU \bS \bM\| \leq 1.$ We use \cref{lem:help1} for each term to see that 
    \begin{align*}
        \P(\|\bM^* \bV^*\bX^* \bZ \bV \bM\| \geq t/3) 
        &\leq C^r \exp\left( - \frac{ct^2}{\tau^2} (\sigma_r^2+p_1) \right), \\
        \P(\|\bM^* \bV^*\bZ^* \bX \bV \bM\| \geq t/3)
        &\leq C^r \exp\left( - \frac{ct^2}{\tau^2} (\sigma_r^2+p_1) \right), \\
        \P(\|\bM^* \bV^* (\bZ^*\bZ- p_1 \id_{p_2})\bV \bM\|\geq t/3)
        &\leq C^r \exp\left( -c\min\left\{ \frac{t}{3\tau^2}(\sigma_r^2+p_1), \, \frac{t^2}{9p_1 \tau^4} (\sigma_r^2+p_1)^2 \right\} \right) \\
        &\leq C^r \exp\left( - c_\tau (\sigma_r^2+p_1) \min\left\{t,t^2\right\} \right).
    \end{align*}
    This now implies
	\begin{equation}\label{eq:help3}
    \begin{split} 
	&\P(\|\bM^* \bY_1^* \bY_1 \bM-\id_r\|\geq t) \\
	&\leq \exp \left(Cr - c_\tau (\sigma_r^2 + p_1) \min\{t,t^2\} \right).
    \end{split} 	
	\end{equation}
	Combining these inequalities yields \eqref{eq:CZvar1}. \\
	
	\noindent \underline{Proof of \eqref{eq:CZvar2}}. 
    By the Eckart–Young–Mirsky theorem, 
	\begin{align*}
	\sigma_{r+1}(\bY)&=\min_{\operatorname{rank}(\bA)\leq r}\|\bY-\bA\| \\
    &= \min_{\operatorname{rank}(\bA)\leq r}\|\bX+\bZ-\bA\|\leq \|\bZ\|,    
	\end{align*}
	where the final inequality follows by taking $\bA=\bX$. Hence, using that $\E (\bZ^* \bZ) = p_1 \id_{p_2}$, we find 
	$$\sigma_{r+1}^2(\bY)\leq \|\bZ\|^2 =  \sigma_1(\bZ^* \bZ) \le p_1 + \|\bZ^* \bZ -\E (\bZ^* \bZ)\|.$$  	
	We now use \cref{lem:help1} to get  
	\begin{align*}
	    \P(p_1^{-1}\|\bZ^* \bZ -\E (\bZ^* \bZ)\|>t)
	    &\leq C^{p_2}\exp(-c_\tau p_1 \min\{t,t^2\}). 
	\end{align*}
	Hence, with probability at least $1-\exp(Cp_2-c_\tau p_1 \min\{t,t^2\})$, we have $\|\bZ^* \bZ -\E (\bZ^* \bZ)\|\leq p_1 t$ and in particular
	$$
	\sigma_{r+1}^2(\bY) \leq p_1(1+t). 
	$$
	
	\noindent \underline{Proof of \eqref{eq:CZvar3}}. We start with the calculation
	\begin{align*}
	\|P_{\bY \bV} (\bY \bV_\perp)\|
	&=\| (\bY_1 \bM) (\bY_1 \bM)^{\dagger} (\bY \bV_\perp)\| \\
	&\leq \sigma_r^{-1} (\bY_1 \bM) \|\bM^* \bV^* \bY^* \bY \bV_\perp\|. 
	\end{align*}
	We treat these two terms separately. For the first, note that $\sigma_r^2 (\bY_1 \bM)=\sigma_r(\bM^* \bY_1^* \bY_1 \bM)\geq 1 -\|\bM^* \bY_1^* \bY_1 \bM-\id_r\|$. Using \eqref{eq:help3} with $t=\frac 12$ and the assumption that $\sigma_r(\bX) \geq C_\tau p_2$ hence $\sigma_r(\bX) \geq C_\tau r$, we have 
	\begin{align*}
	&\P\left( \sigma_r^2(\bY_1\bM)< \frac 12 \right) \\
	& \leq  \P\left( \|\bM^* \bY_1^* \bY_1 \bM-\id_r\| > \frac 12 \right) \\
    &\leq \exp(Cr - c_\tau (\sigma_r^2 + p_1))\\
	&\leq \exp(-c_\tau (\sigma_r^2 + p_1)).     
	\end{align*}
	For the second term, since $\bX \bV_\perp = \boldsymbol{0}$ and $\bV^*\bV_\perp=\boldsymbol{0}$, we see that 
	\begin{align*}
	&\bM^* \bV^* \bY^* \bY \bV_\perp \\
	&=\bM^* \bV^* \bX^* \bZ \bV_\perp + \bM^* \bV^* (\bZ^* \bZ-p_1 \id_{p_2}) \bV_\perp.    
	\end{align*}
    Recall that $\|\bX\bV \bM\|\leq 1$ and $\|\bV \bM\|\leq (\sigma_r^2+p_1)^{-1/2}$. Note that $\|\bV_\perp\|\leq 1$. By \cref{lem:help1},
    \begin{align*}
        \P(\|\bM^* \bV^* \bX^* \bZ \bV_\perp\|>t/2)
        &\leq C^{p_2} \exp\left(-\frac{c t^2}{\tau^2 }(\sigma_r^2+p_1)\right), \\
        \P(\|\bM^* \bV^* (\bZ^* \bZ-p_1 \id_{p_2}) \bV_\perp\|>t/2)
        &\leq C^{p_2} \exp\left(-c \min\left\{ \frac t{\tau^2} \sqrt{\sigma_r^2+p_1}, \frac{t^2}{p_1\tau^4} (\sigma_r^2+p_1)\right\}\right) \\
        &\leq C^{p_2} \exp\left(-c \min\left\{ \frac t{\tau^2} \sqrt{\sigma_r^2+p_1}, \frac{t^2}{\tau^4}\right\}\right). 
    \end{align*}
    Thus, we get
	\begin{align*}
	&\P(\|\bM^* \bV^* \bY^* \bY \bV_\perp\|> t) \\
	&\leq \exp\left(C p_2-c_\tau \min\left\{t^2, t \sqrt{\sigma_r^2+p_1}\right\}\right).    
	\end{align*}
	To finish the proof, we note that 
    \begin{align*}
        &\P(\|P_{\bY \bV} (\bY \bV_\perp)\|>t) \\
        &\leq \P(\sigma_r^{-1} (\bY_1 \bM) \|\bM^* \bV^* \bY^* \bY \bV_\perp\|>t) \\
        &\leq \P\left(\sigma_r^{-1} (\bY_1 \bM) \|\bM^* \bV^* \bY^* \bY \bV_\perp\|>t, \, \sigma_r^{-2}(\bY_1 \bM)\leq 2\right) \\
        &\qquad + \P\left(\sigma_r^{-2}(\bY_1 \bM)> 2\right) \\
        &\leq \P\left(\|\bM^* \bV^* \bY^* \bY \bV_\perp\|> \frac t {\sqrt 2} \right)+\P\left(\sigma_r^{-2}(\bY_1 \bM)> 2\right) \\
        &\leq \exp\left(C p_2-c_\tau \min\left\{t^2, t \sqrt{\sigma_r^2+p_1}\right\}\right)+\exp(-c_\tau (\sigma_r^2 + p_1)).
    \end{align*}
\end{proof}

We are now ready to prove the main result of this section.

\begin{proof}[Proof of Theorem~\ref{thm:CZvarHP}]
    Thanks to the trivial estimate $\|\sin(\bV,\hat \bV)\|\leq 1$, it suffices to prove the result if 
    \begin{equation}
        \label{eqn:sigma4LB}
        \sigma_r^4(\bX)\geq {C} \alpha \, p_2 \, (\sigma_r^2(\bX)+p_1).
    \end{equation}
    Observe that this holds if and only if 
    $$(\sigma_r^2(\bX)-\sigma_+)(\sigma_r^2(\bX)-\sigma_-)\geq 0,$$
    where 
    $$\sigma_{\pm} = \frac{1}{2}\left({C} \alpha p_2 \pm \sqrt{{C^2} \alpha^2 p_2^2 + 4 {C} \alpha p_1p_2}\right).$$
    Since $\sigma_r^2(\bX)\geq 0\geq \sigma_-$, we see that 
    \begin{equation}
        \label{eqn:sigma2LB}
        \sigma_r^2(\bX)\geq \sigma_+\geq {C} \alpha p_2.
    \end{equation}
    By \cite[Proposition 1]{cai2018rate}, if $\bV_{\perp}\in \mathbb{F}^{p_2\times(p_2-r)}$ is such that $[\bV \ \bV_{\perp}]$ is an orthogonal matrix  
    and $\sigma_{r}(\bY\bV)>\sigma_{r+1}(\bY)$, then 
    $$\|\sin(\bV,\hat \bV)\|^2 \leq \min \left\{ 1,\frac {\sigma_r^2(\bY\bV)\|P_{\bY\bV}(\bY\bV_{\perp})\|^2}{(\sigma_r^2(\bY\bV)-\sigma_{r+1}^2(\bY))^2}\right\}.$$ 
    Throughout this proof, $C',c'$, etc. are absolute constants, while $C_\tau,c_\tau$ are constants that only depend on $\tau$. The value of these constants may change from one line to another.
    For $t=\sigma_r^2(\bX)/ {4} (\sigma_r^2(\bX)+p_1)$ consider the event
    $$A_t=\{\sigma_r^2(\bY\bV)\geq (\sigma_r^2(\bX)+p_1)(1-t)\}.$$
    Using \eqref{eq:CZvar1} and $\min\{t,t^2\}=t^2$ as $0<t\leq 1$, we find
    \begin{align*}
        \P(A_t^c) 
        &\leq \exp({C'} r-(\sigma_r^2(\bX)+p_1)\min\{t,t^2\}) \\
        & \leq \exp\left({C'} r- c_\tau \frac{\sigma_r^4(\bX)}{\sigma_r^2(\bX)+p_1}\right) \\
        & \leq \exp(-c_\tau\alpha p_2)
    \end{align*}
    by \eqref{eqn:sigma4LB}, provided that $\alpha_\tau$ is large enough.
    For $s=\sigma_r^2(\bX)/({4} p_1)$ consider
    $$B_s=\{\sigma_{r+1}^2(\bY)\leq p_1(1+s)\}.$$
    By \eqref{eq:CZvar2}, 
    \begin{align*}
        \P(B_s^c) 
        &\leq \exp({C'} p_2- c_\tau p_1\min\{s,s^2\}) \\
        & = \exp\left({C'} p_2-c_\tau p_1\min\left\{\frac{\sigma_r^2(\bX)}{p_1},\frac{\sigma_r^4(\bX)}{p_1^2}\right\}\right) \\
        & \leq \exp\left({C'} p_2- c_\tau
        \frac{\sigma_r^4(\bX)}{\sigma_r^2(\bX)+ p_1}\right) \\
        &\leq \exp(-c_\tau\alpha p_2)
    \end{align*}
    by \eqref{eqn:sigma4LB}, provided that $\alpha_\tau$ is large enough.
    Finally, for $u=\sqrt{\alpha p_2}$ consider
    $$C_u = \{\|P_{\bY\bV}(\bY\bV_{\perp})\|\leq u\}.$$
    By \eqref{eq:CZvar3},
    \begin{align*}
        \P(C_u^c)
        &\leq \exp\left({C'} p_2 - c_\tau\min\left\{\alpha p_2,\sqrt{\alpha p_2} \sqrt{\sigma_r^2(\bX)+p_1}\right\} \right) \\
        &\qquad + \exp\left(-c_\tau( \sigma_r^2(\bX)+p_1)\right) \\
        &\leq \exp(-c_\tau\alpha p_2)
    \end{align*}
    if $\alpha_\tau$ is large enough, as $\sigma_r^2(\bX)\geq {C} \alpha p_2$ by \eqref{eqn:sigma2LB}. 
    {Upon the events $A_t,B_s,C_u$ we now find that $\sigma_{r}(\bY\bV)>\sigma_{r+1}(\bY) \ge 0$ and,}
    using that the map 
    $$x\mapsto {\frac{x}{(x-y)^2}}$$
    is decreasing on $({y},\infty)$ for any fixed {$y \ge 0$} and 
    $$y\mapsto {\frac{x}{(x-y)^2}}$$
    increasing on $(-\infty,{x})$ for any fixed {$x \ge 0$, we obtain} 
    \begin{align*}
        & \frac {\sigma_r^2(\bY\bV)\|P_{\bY\bV}(\bY\bV_{\perp})\|^2}{(\sigma_r^2(\bY\bV)-\sigma_{r+1}^2(\bY))^2} \\
        & \qquad \leq \frac{(\sigma_r^2(\bX)+p_1)(1-t)\alpha p_2}{((\sigma_r^2(\bX)+p_1)(1-t)-p_1(1+s))^2} \\
        & \qquad \lesssim \frac{\alpha p_2(\sigma_r^2(\bX)+p_1)}{\sigma_r^4(\bX)} 
    \end{align*}
    as 
    \begin{align*}
        &(\sigma_r^2(\bX)+p_1)(1-t)-p_1(1+s) \\
        &= \sigma_r^2(\bX)(1-t) {-} p_1(s+t) 
        = \frac{1}{2} \sigma_r^2(\bX).
    \end{align*}
\end{proof}

\subsection{Proof of \cref{thm:subspacerectangletriangle}(ii) and heavy-tailed generalization}
\label{sec:CZvariationII}

\cref{thm:subspacerectangletriangle}(ii) follows by combining Theorem~\ref{thm:subspacerectangletriangle}(i) with \cref{lem:lambdarLB} below for $r = s$. Just note that $\mathrm{range}(\bX_n) = \mathrm{range}(\SigmaX)$ a.s.\ since $n \geq s$ and $\bx \sim \calC\calN(\0,\SigmaX)$.

\begin{lemma}
    \label{lem:lambdarLB}
    There exist absolute constants $c_1,c_2,c_3 > 0$ such the following holds. Suppose that $\bx_1,\dots,\bx_n \overset{i.i.d.}{\sim} \calC\calN(\0,\SigmaX)$. If 
    {
    $$n \ge c_1 \frac{\sum_{j=1}^r \lambda_j(\SigmaX)}{\lambda_r(\SigmaX)},$$
    then with probability at least $1-2e^{-c_2 n}$}
    $$\lambda_r(\tfrac{1}{n} \bX_n\bX_n^*)\geq c_3 \lambda_r(\SigmaX).$$
\end{lemma}

We extract the following technical observation from the proof of Lemma \ref{lem:lambdarLB}.

\begin{lemma}
\label{lem:Mendelson}
    There exist absolute constants $c_1,c_2,c_3 > 0$ such the following holds. Suppose that $\bg_1,\dots,\bg_n \overset{i.i.d.}{\sim} \calC\calN(\0,\id_p)$, let $\SigmaX \in \C^{p\times p}$ be a covariance matrix, and let $\bU \in \C^{p\times r}$ have orthonormal columns that span the leading eigenspace of $\SigmaX$. If
    {
    $$n \ge c_1 \frac{\sum_{j=1}^r \lambda_j(\SigmaX)}{\lambda_r(\SigmaX)},$$
    then with probability at least $1-2e^{-c_2 n}$}, 
    $$\inf_{\|\bv\|_2=1}\frac{1}{n}\sum_{k=1}^n \left| \langle \bg_k,  \SigmaX^{1/2} {\bU} \bv\rangle \right|^2 \geq c_3 \lambda_r(\SigmaX).$$
\end{lemma}

\begin{proof}
   For any {$\bv \in \C^p$}, we note that 
   \begin{align*}
       \frac{1}{n}\sum_{k=1}^n \left| \langle \bg_k,  \SigmaX^{1/2} {\bU} \bv\rangle \right|^2
       \ge {\frac{1}{n}\sum_{k=1}^n \Re ( \langle \bg_k,  \SigmaX^{1/2} {\bU} \bv\rangle )^2}
       = {\frac{1}{n}\sum_{k=1}^n \langle \tilde\bg_k, \tilde\bv \rangle^2},
   \end{align*}
   where
   \begin{align*}
       {\tilde\bv = \begin{pmatrix}
           \Re(\SigmaX^{1/2} \bU \bv) \\
           \Im(\SigmaX^{1/2} \bU \bv)
       \end{pmatrix}
       \in \R^{2p}}
       \quad \text{and} \quad
       \tilde\bg_k = \begin{pmatrix}
           \Re(\bg_k) \\ \Im(\bg_k) 
       \end{pmatrix}
       \in \R^{2p}.
   \end{align*}
   By Mendelson's small ball method, see (the proof of) \cite[Theorem 1.5]{KoM15}, with probability at least $1-2e^{-2t^2}$, 
   \begin{align}
   \label{eq:Mendelson}
       \inf_{\|\bv\|_2=1}\frac{1}{n}\sum_{k=1}^n {\langle \tilde\bg_k, \tilde\bv \rangle^2} \geq \xi^2\left(Q_{2\xi}-\frac{4W_n}{\xi \sqrt{n}}-\frac{t}{\sqrt{n}}\right),
   \end{align}
    where 
    $$Q_{2\xi}=\inf_{\|\bv\|_2=1} \P({|\langle \tilde\bg_k, \tilde\bv \rangle |} \geq \xi),$$
    $\xi > 0$ is sufficiently small such that $Q_{2\xi} > 0$, and 
    \begin{align}
    \label{eq:W}
        W_n &= \E \sup_{\|\bv\|_2=1} \left| \frac{1}{\sqrt{n}} \sum_{k=1}^n \varepsilon_k {\langle \tilde\bg_k, \tilde\bv \rangle} \right|
        = {\E \sup_{\|\bv\|_2=1} \left| \Re \left( \left\langle \frac{1}{\sqrt{n}} \sum_{k=1}^n \varepsilon_k \bg_k, \SigmaX^{1/2} \bU \bv \right\rangle \right) \right|} \notag \\
        &= {\E \sup_{\|\bv\|_2=1} \left| \Re \left( \left\langle \bg, \SigmaX^{1/2} \bU \bv \right\rangle \right) \right|}
        \le {\E \| \bU^* \SigmaX^{1/2} \bg \|_2 },
    \end{align}
    where $\varepsilon_k$ are i.i.d.\ Rademacher variables {and $\bg \sim \calC\calN(\0,\id_p)$}.

    {Since $\langle \tilde\bg_k, \tilde\bv \rangle \sim \calN(0,\| \tilde\bv \|_2^2)$ and $\| \tilde\bv \|_2 = \| \SigmaX^{1/2} \bU \bv \|_2$, we know that $(\E |\langle \tilde\bg_k, \tilde\bv \rangle|)^2 \le \E |\langle \tilde\bg_k, \tilde\bv \rangle|^2 \le (c \E |\langle \tilde\bg_k, \tilde\bv \rangle|)^2$, for an absolute constant $c > 0$, and that $(\E |\langle \tilde\bg_k, \tilde\bv \rangle|^2)^\frac{1}{2} = \| \tilde\bv \|_2 \geq \sqrt{ \lambda_r(\SigmaX)}$ for all $\|\bv\|_2=1$.} Using these observations and the Paley-Zygmund inequality, and choosing $\xi = c' \sqrt{\lambda_r(\SigmaX)}$ for sufficiently small $c' > 0$, yield
    \begin{align*}
        Q_{2\xi} & \geq \inf_{\|\bv\|_2=1} \frac{(\E {|\langle \tilde\bg_k, \tilde\bv \rangle|} -\xi)^2}{\E {|\langle \tilde\bg_k, \tilde\bv \rangle|}^2} 
        \ge \left( c^{-1} - \frac{\xi}{c \cdot \inf_{\|\bv\|_2=1} (\E |\langle \tilde\bg_k, \tilde\bv \rangle|^2)^\frac{1}{2}} \right)^2 
        \ge (2c)^{-2} > 0.
    \end{align*}
    {Moreover, by \eqref{eq:W}
    \begin{align*}
        W_n & \leq (\E\|\bU^*\SigmaX^{1/2}\bg\|_2^2)^{1/2}
        = \big( \E\, \tr(\bU^*\SigmaX^{1/2}\bg \bg^*\SigmaX^{1/2}\bU) \big)^{1/2} \\
        &= \big( \tr(\bU^*\SigmaX \bU) \big)^{1/2} 
        = \left(\sum_{j=1}^r \lambda_j(\SigmaX)\right)^{1/2} 
    \end{align*}
    By setting $\xi = \sqrt{\lambda_r(\SigmaX)}$ and $t = \sqrt{c_2 n}$, the claim now follows from \eqref{eq:Mendelson} by our assumption on $n$ if we choose $c_1$ to be sufficiently large and $c_2$ to be sufficiently small with respect to $c$.
    }
\end{proof}

\begin{proof}[Proof of Lemma \ref{lem:lambdarLB}]
    We prove the lemma for $\F = \C$. The case $\F = \R$ follows by simplification of the argument. Let $\bg_1,\ldots,\bg_n$ be i.i.d.\ complex standard Gaussian, so that $\bx_k\sim \SigmaX^{1/2}\bg_k$ for all $k\in[n]$. By the Courant-Fisher minimax principle, 
    $$\lambda_r(\tfrac{1}{n} \bX_n\bX_n^*) = \max_{\substack{M\subset \C^p: \\ \textnormal{dim}(M)=r}} \; \min_{\substack{\bu\in M: \\ \|\bu\|_2=1}} \left \langle \tfrac{1}{n} \bX_n\bX_n^*\bu,\bu\right\rangle.$$
    Thus, if we fix any $r$-dimensional subspace $M$ and let $\bU_M \in \C^{p\times r}$ contain an orthonormal basis of $M$, we find
    $$\lambda_r(\tfrac{1}{n} \bX_n\bX_n^*)\geq \min_{\|\bv\|_2=1} \left \langle \tfrac{1}{n} \bX_n\bX_n^*\bU_M \bv,\bU_M \bv\right\rangle = \min_{\|\bv\|_2=1} \frac{1}{n}\sum_{k=1}^n \left| \langle \bg_k,  \SigmaX^{1/2}\bU_M \bv\rangle \right|^2.$$ 
    The result now follows from Lemma \ref{lem:Mendelson}.
\end{proof}

By inspecting the proof of Lemma \ref{lem:lambdarLB}, it is straightforward to extend the result to heavy-tailed distributions which only satisfy a mild $L^1$-$L^2$-equivalence, i.e., there exists $c > 0$ such that
\begin{align}
\label{eq:L1-L2equivalence}
    \E\, \Re(\langle \bx, \bv \rangle)^2 \le c (\E\, \Re(\langle \bx, \bv \rangle) )^2,
\end{align}
for any $\bv \in \C^p$. From this the following more general version of Theorem \ref{thm:subspacerectangletriangle}(ii) can be deduced.
\begin{theorem}
\label{thm:subspacerectangletriangle_heavytailed}
    Suppose Assumptions \ref{assump:main} and \ref{assump:noise2} hold and that $\rank(\SigmaX) = s$. There are absolute constants {$C_1,C_2 > 0$}, and constants $c,\alpha_\tau>0$ depending only on $K$ such that the following holds. Let $\bU,\bU_n^\triangle \in \O^{p\times s}$ denote orthonormal bases for the leading left singular spaces of $\SigmaX$ and $\SIGMADITHtriangular_n$. Assume that $\bx_1,\dots,\bx_n \overset{i.i.d.}{\sim} \bx$ where $\bx$ has mean zero and satisfies \eqref{eq:L1-L2equivalence}, and denote the Rademacher complexity
    \begin{align*}
        W_n = \E \sup_{\| \bv \|_2 = 1} \left| \frac{1}{\sqrt{n}} \sum_{k=1}^n \varepsilon_k \Re(\langle \bx_k, \bU\bv \rangle) \right|,
    \end{align*}
    where $\varepsilon_k$ are i.i.d.\ Rademacher variables.
    For any $\alpha \ge \alpha_\tau$, 
    we have that if {$n \ge C_2 \frac{W_n^2}{\lambda_s(\SigmaX)}$, then with probability at least $1-3e^{-c \min\{\alpha p, n \}}$}
    \begin{equation}
        \dist(\bU_n^\triangle,\bU)
        \leq \min \left\{ 1, \, C_1 \sqrt{\frac {\nu^2+\mu^2}{\lambda_s(\SigmaX)} } \left( 1 + \sqrt{ \frac{\nu^2+\mu^2}{\lambda_s(\SigmaX)} } \right) \sqrt{\frac{\alpha p}{n}} \right\}.
    \end{equation}
\end{theorem}

\subsection{Proof of \cref{thm:espritrectangular}}
\label{proof:espritrectangular}

We first recall the following lemma that summarizes \cite[Lemma X.1]{li2022stability}, \cite[Lemma 2]{li2020super}, and \cite[Lemma X.2]{li2022stability}. It shows that any (deterministic) estimator $\hat\bU$ of $\bU$ produces an estimate $\hat\btheta$ of the true frequencies $\btheta$ via the ESPRIT algorithm.

\begin{lemma}
	\label{lem:mdsintheta}
	Suppose $p > s$. Then  
	$$
	\md(\hat\btheta,\btheta) 
	\leq \min\left\{ \frac 12, \, \frac{ 4^{s+2} s^{3/2} \sqrt p  }{\sigma_s(\bPhi)} \dist(\hat\bU,\bU) \right\}. 
	$$
	If additionally,
	$
	\dist(\hat\bU,\bU)
	\leq \displaystyle \frac{\sigma_s^2(\bPhi) \Delta}{4^{s+2} s^2 p}, 
	$
	then
	$$
	\md(\hat\btheta,\btheta) 
	\leq \min \left\{ \frac 12, \, 4^{s+2} \dist(\hat\bU,\bU) \right\}.
	$$
\end{lemma}

The first inequality of this lemma holds unconditionally. The second one requires an assumption that the perturbation of $\bU$ is sufficiently small, in which case, we obtain an improved estimate for the perturbation of $\btheta$. 

We will use \cref{lem:mdsintheta} with $\bU_n^\square$ acting as an estimator for $\bU$. To bound $\dist(\bU_n^\square,\bU)$ via \cref{thm:subspacerectangle}, we need to show that \cref{assump:specestimationproblem} can be put into the framework of \cref{assump:main}. 

\begin{lemma}
	\label{lem:model}
	If \cref{assump:specestimationproblem} holds, then \cref{assump:main} is satisfied with $\bx_k:= \bPhi\ba_k$ and $\by_k=\bPhi\ba_k+\be_k$ for each $k\in [n]$, and furthermore, $\SigmaX$ has rank $s$. We also have the inequality 
	\begin{equation}
		\label{eq:Sighelp1}
		\lambda_s(\SigmaX)\geq \sigma_s^2(\bPhi) \lambda_s(\SigmaA). 
	\end{equation} 
\end{lemma}

\begin{proof}
	In the case of stochastic amplitudes, $\ba_1,\dots,\ba_n$ are independent and $K/2$-subgaussian. Then $\bx_1,\dots,\bx_n$ are independent and $K$-subgaussian as well and we have    
	$$
	\SigmaX:=\E \bx\bx^*=\bPhi (\E \ba\ba^*) \bPhi^* =  \bPhi \SigmaA \bPhi^*.
	$$
	Since $\bPhi$ is Vandermonde, since $\btheta$ consists of distinct elements in $\R/\Z$, and since $s\leq p$, we see that $\bPhi$ has rank $s$. Due to the assumption that $\SigmaA$ is of full rank, this implies that $\SigmaX$ has rank $s$. This yields \eqref{eq:Sighelp1} for the stochastic case. 
	
	In the case of deterministic amplitudes $\ba_1,\dots,\ba_n$, we see that $\bx_1,\dots,\bx_n$ are deterministic vectors since $\bPhi$ is fixed. We have
	$$
	\SigmaX
	:= \frac 1 n \sum_{k=1}^n \bx_k\bx_k^*
	=\bPhi \Big( \frac 1 n \sum_{k=1}^n \ba_k\ba_k^* \Big) \bPhi^*
	=\bPhi \SigmaA \bPhi^*.
	$$
	Since both $\bPhi$ and $\SigmaA$ have rank $s$, this calculation establishes that $\SigmaX$ has rank $s$ as well. This yields \eqref{eq:Sighelp1} for the deterministic case. 
\end{proof}

\begin{proof}[Proof of \cref{thm:espritrectangular}]
    We start with the proof by verifying the conditions of \cref{thm:subspacerectangle} are fulfilled. To estimate the two quantities $C_\infty$ and $\Cop$ defined in \eqref{eq:Cy_defs}, we use their simplifications in \cref{rem:Cop}. Regardless of $\btheta$, the entries of $\bPhi$ have unit modulus. For deterministic amplitudes, we note that
    \begin{align*}
        C_\infty
        &= \max_{k\in [n]} \|\bPhi \ba_k \ba_k^* \bPhi^*\|_\infty^2 + \nu^2 \\
        &\leq \max_{k\in [n]} \|\ba_k\ba_k^*\|_1 + \nu^2= \gamma_1+\nu^2, \\
        \Cop 
        &= \max_{k\in [n]} \|\bPhi \ba_k\|_2^2 + \nu^2 \\
        &\leq \sigma_1^2(\bPhi) \max_{k\in [n]} \|\ba_k\|_2^2 + \nu^2= \gamma_2^2\sigma_1^2(\bPhi)+\nu^2. 
    \end{align*}
    For stochastic amplitudes, we have the inequalities 
    \begin{align*}
        C_\infty
        &= \|\SigmaX\|_\infty^2 + \nu^2
        = \|\bPhi \SigmaA \bPhi^* \|_\infty^2 + \nu^2 \\
        &\leq \|\SigmaA\|_1 + \nu^2= \gamma_1+\nu^2, \\
        \Cop
        &=\|\SigmaX\|+\nu^2 
        =\| \bPhi \SigmaA \bPhi^*\|+\nu^2 \\
        &\leq \sigma_1^2(\bPhi) \lambda_1(\SigmaA) + \nu^2= \gamma_2^2\sigma_1^2(\bPhi)+\nu^2. 
    \end{align*}
    Using the assumed lower bound for $\lambda$, we see that the conditions of \cref{thm:subspacerectangle} are fulfilled. We use this theorem for $t \log(p)$ in place of $t$ and \eqref{eq:Sighelp1}, to see that whenever $n\geq (t+1)\lambda^2 p \log(p)$, with probability at least $1-p^{-t}$,
    \begin{equation}
    	\label{eq:probboundhelp1}
    	\dist(\bU_n^\square,\bU)
    	\lesssim_K \frac{\gamma_2\sigma_1(\bPhi) + \nu+\lambda}{\sigma_s^2(\bPhi) \lambda_s(\SigmaA)} \sqrt{\frac{(1+t)p \log (p)}{n}}.
    \end{equation}
    Let $\calE$ be the event that inequality \eqref{eq:probboundhelp1} holds, which occurs with probability at least $1-p^{-t}$. Due to the second condition on $n$ in \eqref{eq:ncondition}, { we have under the event $\calE$ that}
    $$
    \dist(\bU_n^\square,\bU)
    \leq \frac{\sigma_s^2(\bPhi) \Delta}{4^{s+2} s^2 p}. 
    $$
    This enables us to use the second statement of \cref{lem:mdsintheta} to get 
    $$
    \md(\btheta_n^\square,\btheta)
    \leq \min \left\{ \frac 12, \, 4^{s+2} \dist(\bU_n^\square,\bU) \right\}. 
    $$
    Combining these observations completes the proof.
\end{proof}

\subsection{Proof of \cref{thm:esprittriangular}}
\label{proof:esprittriangular}
By \cref{lem:model}, we see that \cref{assump:main} holds and $\rank(\SigmaX)=s$. This enables us to use \cref{thm:subspacerectangletriangle}. Let $\bU_n^\triangle\in \O^{p\times s}$ be an orthonormal basis for  the leading left singular space of $\bQ_n^\triangle$. Consider the event $\calE_1$, where 
    \begin{equation}
    \label{eq:trievent1}
	\dist(\bU_n^\triangle,\bU)
	\leq \min \left\{ 1, \, C \sqrt{\frac {\nu^2+\mu^2}{\sigma_s^2(\bPhi)\lambda_s(\SigmaA)}}\sqrt{\frac{\alpha p}{n}} \right\}. 
	\end{equation}
    Also consider the event $\calE_2$,
    \begin{equation}
        \label{eq:trievent2}
        \dist(\bU_n^\triangle,\bU)
	\leq \frac{\sigma_s^2(\bPhi) \Delta}{4^{s+2} s^2 p}. 
    \end{equation}
	
	For part (i), we use \cref{thm:subspacerectangletriangle}  part (i), assumption \eqref{eq:noisecond1}, and inequality \eqref{eq:Sighelp1} to see that event $\calE_1$ holds with probability at least $1-e^{-c\alpha p}$. Under event $\calE_1$ and together with assumption \eqref{eq:ncond1}, we see that event $\calE_2$ holds. This enables us apply the second statement in \cref{lem:mdsintheta} to complete the proof of part (i).

    {
	For part (ii), we note that $\bx:=\bPhi \ba \sim \calC\calN(\0,\SigmaX)$. Letting $C_2$ be the constant in \cref{thm:subspacerectangletriangle}, we observe that 
    \begin{align*}
        C_2\beta \frac{\sum_{j=1}^s \lambda_j(\SigmaX)}{\lambda_s(\SigmaX)} 
        &= C_2 \beta \frac{\tr(\bPhi \SigmaA \bPhi^*)}{\sigma_s^2(\bPhi)\lambda_s(\SigmaA)}
        =C_2 \beta \frac{\tr(\bPhi^*\bPhi\SigmaA)}{\sigma_s^2(\bPhi)\lambda_s(\SigmaA)} \\
        &\leq C_2 \beta \frac{\sqrt{\tr((\bPhi^*\bPhi)^2) \tr(\SigmaA^2)}}{\sigma_s^2(\bPhi)\lambda_s(\SigmaA)} 
        =C_2 \beta \Bigg(\sum_{j=1}^s \kappa_j^4(\bPhi)\Bigg)^{1/2} \Bigg(\sum_{j=1}^s \kappa_j^2(\SigmaA)\Bigg)^{1/2}.
    \end{align*}
    In view of assumption \eqref{eq:ncond2} (where we set $c_1$ to be $C_2$), we can use \cref{thm:subspacerectangletriangle} part (ii).} We also use  \eqref{eq:noisecond1} and inequality \eqref{eq:Sighelp1} in the resulting expression. So with probability at least $1-3e^{-c\min\{\alpha p,\beta r\}}$ event $\calE_1$ holds. Employing condition \eqref{eq:ncond2} again, we see that $\calE_2$ holds whenever $\calE_1$ occurs.

\section*{Acknowledgments}

WL is supported by NSF-DMS Award \#2309602 and a Cycle 55 PSC-CUNY grant.

\bibliographystyle{plain}
\bibliography{MSSreference,SRlimitFourierbib}

\end{document}